\documentclass[11pt]{article}

\usepackage{amsmath,amssymb,graphicx,amsthm,amsfonts, hyperref}

\usepackage{newpxmath}
\usepackage{newpxtext}
\usepackage{setspace}
\usepackage{latexsym,amsfonts}

\usepackage{mathtools}
\usepackage{booktabs,xcolor}
\usepackage[margin=2cm ]{geometry}

\usepackage{enumitem}
\usepackage{scalerel}

\newcommand{\cref}{\ref}

\usepackage{environ}
\usepackage{cancel}
\usepackage{graphicx}

\usepackage{xcolor}

\usepackage[toc,page]{appendix}

\usepackage{etoolbox}

\newtheorem{theorem}{Theorem}[section]

\newtheorem{lemma}[theorem]{Lemma}
\newtheorem{claim}[theorem]{Claim}
\newtheorem{corollary}[theorem]{Corollary}

\newtheorem*{question*}{Question}

\newtheorem{definition}[theorem]{Definition}

\newtheorem*{definition*}{Definition}

\usepackage{mathtools}
\DeclarePairedDelimiterX\setc[2]{\{}{\}}{\,#1 \;\delimsize\vert\; #2\,}

\DeclareMathOperator{\sgn}{sign}

\DeclareMathOperator*{\EE}{\mathbb{E}}

\DeclareMathOperator{\Ber}{Ber}

\newcommand{\RR}{\mathbb{R}}
\newcommand{\PP}{\mathbb{P}}

\newcommand{\cB}{\mathcal{B}}

\newcommand{\cD}{\mathcal{D}}

\newcommand{\cN}{\mathcal{N}}

\newcommand{\cS}{\mathcal{S}}
\newcommand{\cT}{\mathcal{T}}

\newcommand{\eps}{\varepsilon}

\title{Adversarially-Robust Inference on Trees via Belief Propagation}
\author{Samuel B. Hopkins\thanks{Massachusetts Institute of Technology, 77 Massachusetts Ave, Cambridge, MA 02139. SBH is supported by NSF CAREER award No. 2238080 and MLA@CSAIL. \texttt{samhop@mit.edu.}} \hspace{2 pt} and Anqi Li\thanks{Department of Pure Mathematics and Mathematical Statistics (Centre for Mathematical Sciences), University of Cambridge, Cambridge CB30WB, England, UK. AL is supported by the Trinity Studentship in Mathematics. Part of this work was completed while AL was an intern at Microsoft Research (New England). \texttt{anqili@mit.edu.}}}

\begin{document}

\maketitle

\begin{abstract}
    We introduce and study the problem of posterior inference on tree-structured graphical models in the presence of a malicious adversary who can corrupt some observed nodes.
    In the well-studied \emph{broadcasting on trees} model, corresponding to the ferromagnetic Ising model on a $d$-regular tree with zero external field, when a natural signal-to-noise ratio exceeds one (the celebrated \emph{Kesten-Stigum threshold}), the posterior distribution of the root given the leaves is bounded away from $\mathrm{Ber}(1/2)$, and carries nontrivial information about the sign of the root.
    This posterior distribution can be computed exactly via dynamic programming, also known as belief propagation.
    
    We first confirm a folklore belief that a malicious adversary who can corrupt an inverse-polynomial fraction of the leaves of their choosing makes this inference impossible.
    Our main result is that accurate posterior inference about the root vertex given the leaves \emph{is} possible when the adversary is constrained to make corruptions at a $\rho$-fraction of randomly-chosen leaf vertices, so long as the signal-to-noise ratio exceeds $O(\log d)$ and $\rho \leq c \eps$ for some universal $c > 0$.
    Since inference becomes information-theoretically impossible when $\rho \gg \eps$, this amounts to an information-theoretically optimal fraction of corruptions, up to a constant multiplicative factor.
    Furthermore, we show that the canonical belief propagation algorithm performs this inference.
\end{abstract}

\tableofcontents

\newpage

\setcounter{page}{1}

\section{Introduction}

Posterior inference is the central problem in Bayesian statistics: given a multivariate probability distribution, often specified by a \emph{graphical model}, and observed values for a subset of the random variables in this distribution, the goal is to infer resulting conditional, or \emph{posterior}, distribution on the unobserved variables.
Bayesian methods allow rich domain knowledge to be incorporated into inference, via prior distributions, and posterior distributions offer an expressive language to describe the uncertainty remaining in unobserved variables given observations.
Observations are not necessarily distributed iid conditioned on unobserved variables -- a graphical model may specify a complex joint distribution among observed variables.

How robust are posterior inferences to corruptions or errors in observed data?
Or, to misspecifiation of the underlying probabilistic  model?
The rapidly developing field of robust estimation addresses similar questions in the \emph{frequentist} setting where the goal is to construct point estimators which retain provable accuracy guarantees when a small fraction of otherwise iid samples have been maliciously corrupted \cite{diakonikolas2023algorithmic}.
But, to our knowledge, relatively little attention has been paid to questions of robustness to adversarial (i.e. non-probabilistic) contamination in posterior inference.

We study robustness of posterior inference to adversarial corruption and model misspecification for a protoypical family of graphical models: broadcast processes on trees (see \cite{evans2000broadcasting} and references therein).
\begin{definition}[Broadcasting and inference in regular trees]
Given a depth-$t$ tree with vertices $V$ where each non-leaf node has $d$ children, we consider the following joint probability distribution on $\{ \pm 1\}^V$.
First, assign the root vertex a uniform draw from $\{ -1,+1\}$.
Then, recursively assign each child vertex the same sign as its parent with probability $\tfrac{1+\eps}{2}$ and otherwise the opposite sign.
Observing the signs of the leaf nodes $\sigma_L$ in the broadcast tree, our goal is to infer the posterior distribution of the sign $\sigma_R$ of the root.
\end{definition}

The broadcast process on trees is a useful ``model organism'' for our purposes:
\begin{enumerate}
\item Like many graphical models used in practice, it is tree-structured.
This means that posterior inference is algorithmically tractable via dynamic programming, in this setting called ``belief propagation''.
\item Some graphical models exhibit \emph{correlation decay}, meaning that the covariance/mutual information between subsets of random variables decays (exponentially) in the graph distance between the corresponding subsets.
Strong forms of correlation decay (e.g., strong spatial mixing) imply that the distribution on any unobserved node far enough in graph distance to every observed node is insensitive to the values of the observed nodes.
This means that any corruptions in observations won't affect the posterior on the unobserved node, but also means that even without corruptions, no interesting inference can be made about the unobserved node.

As $d$ and $\eps$ increase, the broadcast process exhibits long-range correlations, because the leaf signs $\sigma_L$ collectively carry information about the sign $\sigma_R$ of the root vertex.
In particular, when $d\eps^2 > 1$ (the celebrated \emph{Kesten-Stigum threshold}), this holds even in the limit of infinite tree depth, that is, $\lim_{t \rightarrow \infty} I(\sigma_L ; \sigma_R) > 0$.
Thus, there is nontrival inference to be performed, but by the same token, incorrect observations of the leaf signs could adversely affect the accuracy of this inference.
\end{enumerate}

We introduce and study several models for adversarially-robust inference in a broadcast tree.
In each of these models, a malicious adversary observes the leaf signs $\sigma_L$ resulting from a broadcast process and may flip a subset of them -- the specifics of how this subset is chosen are important, and vary across our models.
The corrupted leaf signs are then passed to an inference algorithm which aims to output the conditional distribution of the root sign \emph{conditioned on the signs of the non-corrupted leaves}.
Crucially, we do not know which subset of leaf signs has been flipped at inference time.

\paragraph{Organization}
In the remainder of this section, we give a high level overview of our results (Section~\ref{sec:results}), discuss how our results can be interpreted in terms of model misspecification (Section~\ref{sec:model-mis}), discuss some related work (Section~\ref{sec:rel-work}) and provide some proof ideas for our main results (Section~\ref{sub-sec:overview}). 


\subsection{Results}\label{sec:results}
The first adversarial model we consider is the simplest and most powerful: for some $\rho > 0$, the \emph{$\rho$-fraction adversary} can choose any $\rho d^t$ leaves (out of $d^t$ leaves in total) and flip their signs.
In this setting we confirm what we believe to be folklore \cite{Polyanskiy2023}:\footnote{However, we are not aware of anywhere it is written in the literature.} for any $\rho > d^{-\Omega(t)}$, this adversary can make the posterior distribution at the root unidentifiable from given the (corrupted) leaf signs.
In what follows, for a random variable $X$, we write $\{ X \}$ for the distribution of $X$, and for an event $E$, we write $\{X \, | \, E\}$ for the distribution of $X$ conditioned on $E$.
We also write $d_{TV}(\cdot,\cdot)$ for the total variation distance between two distributions.
\begin{theorem}[Proof in Section~\ref{sec:info-theory}] 
\label{thm:lower-bound-intro}
    There exists $\eps_0 > 0$ such that for every $\rho$, $d$ and $\eps < \eps_0$, there exists a $\eps^{O(t)}$-fraction adversary $A$ such that if $(\sigma_R,\sigma_L)$ is distributed according to the broadcast process with parameters $d,\eps$,
    \[
    d_{TV}( \{ \sigma_L \, | \, \sigma_R = 1 \}, \{ A(\sigma_L) \, | \, \sigma_R = -1 \}) \leq e^{-\Omega(t)} \, .
    \]
\end{theorem}
An alternative interpretation of Theorem~\ref{thm:lower-bound-intro} is that the Wasserstein distance, with respect to the Hamming metric, between the distributions $\{ \sigma_L \, | \, \sigma_R = 1\}$ and $\{ \sigma_L \, | \, \sigma_R = -1 \}$, decays exponentially with $t$.

Theorem~\ref{thm:lower-bound-intro} implies that no algorithm can reliably distinguish whether $\sigma_R = 1$ or $\sigma_R = -1$, with advantage better than $e^{-\Omega(t)}$ over random guessing, in the presence of an $\eps^{O(t)}$-fraction adversary.
Accurately computing the posterior distribution $\{ \sigma_R \, | \, \sigma_L \}$ and then sampling from it would distinguish these cases with nonvanishing advantage as $t \rightarrow \infty$ (if $\eps^2 d > 1$).
Hence, computing the posterior in the presence of this adversary is impossible.

While this may make it appear that posterior inference in the broadcast tree is inherently non-robust, recent works \cite{MNS16,YP22} tell a contrasting story about a weaker adversary.
The \emph{$\rho$-random} adversary simply flips each leaf sign independently with probability $0 \leq \rho < 1/2$.
It turns out that the posterior at the root vertex can still be accurately recovered in the presence of the $\rho$-random adversary as long as $\rho(\eps,d)$ is held fixed as $t \rightarrow \infty$ \cite{MNS16,YP22} -- this is sometimes call ``robust reconstruction.''
This suggests the question:
\begin{quote}
\begin{center}
\emph{Is posterior inference in the broadcast process possible in the presence of an adversary more malicious than the $\rho$-random adversary?}
\end{center}
\end{quote}

We introduce a \emph{semirandom} adversary, whose power lies in between the worst-case adversary of Theorem~\ref{thm:lower-bound-intro} and the $\rho$-random one.
For the semirandom adversary, the locations of allowed sign flips are random, but the decision whether to make a flip is made adversarially, in full knowledge of all of the leaf signs.
\begin{definition}[$\rho$-semirandom adversary]
Fix $\rho > 0$.
A $\rho$-semirandom adversary receives leaf signs $\sigma_L$ and flips an independent coin $x_u$ for each leaf $u$ which is heads with probability $\rho$.
For each leaf $u$, if $u$'s coin is heads, the adversary may choose to flip the sign $\sigma_u$.
\end{definition}

Our main result shows that when the signal-to-noise ratio $d \eps^2$ exceeds the Kesten-Stigum threshold by a logarithmic factor, there is $\rho(\eps,d) > 0$ such that the distribution of the root vertex can be successfully inferred even in the presence of a $\rho$-semirandom adversary, for large-enough depth $t$. In what follows, we write $(\sigma_R, \sigma_L) \sim \cD_{d, \eps, t}$ to denote $(\sigma_R, \sigma_L)$ distributed according to the broadcast on tree process with parameters $d, \eps$ run up to depth $t$. 


\begin{theorem}[Main theorem, follows from Lemma~\ref{lem:small-eps} and Lemma~\ref{lem:contract-large}]
\label{thm:const-corrupt}
For any $\delta > 0 $ there exists $C$ such that for any $d, \varepsilon$ satisfying $d\varepsilon^2 >C \log \tfrac{d}{1-\eps}$ there exists $\rho_0(\eps) = \Omega(\eps)$ and $t_0(\delta, d)$ such that if $\rho < \rho_0$ and $t > t_0$, for any $\rho$-semirandom adversary $A$, the belief-propagation algorithm \texttt{BP} satisfies
\[
\EE_{\sigma_L \sim \cD_{d, \eps, t}} d_{TV}(\{ \sigma_R \, | \, \sigma_L \}, \texttt{BP}(A(\sigma_L))) \leq \delta \, ,
\]
where the expectation is taken over the broadcast process $\sigma_R, \sigma_L$ as well as the random choice of which vertices the adversary $A$ may choose to corrupt.
\end{theorem}

The algorithm \texttt{BP} in Theorem~\ref{thm:const-corrupt} simply computes the posterior distribution of the root spin $\sigma_R$ as if the leaf spins had been $A(\sigma_L)$ rather than $\sigma_L$, via dynamic programming -- this is the canonical method to compute posterior distributions in tree-structured graphical models \cite{evans2000broadcasting, MNS16}. 
Our analysis of belief propagation borrows techniques from the arguments of \cite{MNS16} for the random-adversary case, but since the semirandom adversary can introduce nasty dependencies among leaf vertices, these arguments are far from transferring immediately.

We can also replace the assumption that the allowed corruptions are in random locations with a natural deterministic assumption on the pattern of allowed corruption locations.
Concretely, if $d\eps^2> C \log \tfrac{d}{1-\eps}$, then for every $c > 0$ there is $k$ such that if the adversary makes at most $c$ corruptions 
 in every height-$k$ subtree of the broadcast tree, we show that our algorithm successfully infers the distribution at the root vertex.
We capture this in Theorem~\ref{thm:well-spread}.

\paragraph{Open question: robustness down to the KS threshold}
Theorem~\ref{thm:const-corrupt} leaves an important open question: is robustness against a semirandom adversary possible for all $d \eps^2 > 1$, as in the random-adversary case?
Or, does the semirandom adversary shift the information-theoretic phase transition from non-recoverability to recoverability of the root spin away from $1$?\footnote{In the related \emph{stochastic blockmodel} setting, a so-called ``monotone'' adversary is known to shift the analogous non-recoverability threshold by a constant factor \cite{moitraperrywein}. In our setting, a monotone adversary would correspond to one who observes the root sign $\sigma_R$ and may flip any leaf sign $\sigma_L$ to agree with $\sigma_R$.}

\subsection{Adversarial Corruption versus Model Misspecification}\label{sec:model-mis}
Adversarial robustness is of course desirable when inference is performed with potentially-corrupted data.
But is it useful beyond protection against malicious data poisoning?

\paragraph{Background: corruptions and misspecification in frequentist statistics} In (frequentist) robust statistics, there is an appealing relationship between adversarial corruption of a subset of otherwise-iid samples and learning/estimation under model misspecification.
Suppose we design a learning algorithm which takes samples from some distribution $D$ in a class of distributions $\mathcal{D}$, of which $1\%$ have been corrupted by an adversary, and successfully learns some $\hat{D}$ which is close to $D$ in total variation distance.

Now, suppose that $\mathcal{D}$ is misspecified, in the sense that it does not contain the ground truth distribution $D$, but only contains some $D' \in \mathcal{D}$ such that $TV(D,D') \leq 0.001$.
Since the adversary could have coupled $D$ and $D'$ to make corrupted samples from $D'$ look as though they are from $D$, then even given samples from $D$ the algorithm must still learn some $\hat{D}$ with small $TV(\hat{D},D)$.

\paragraph{Misspecified Bayesian models}
Adversarially-robust algorithms in our setting are also robust against an appropriate notion of model misspecification.
Concretely, fix a joint probability distribution $\mu(x_0,x_1,\ldots,x_n)$ and random variables $X_0,\ldots,X_n$ jointly distributed according to $\mu$, and consider the posterior inference problem where we observe a joint sample $X_1,\ldots,X_n$ and aim to output the distribution $\{X_0 \, | \, X_1,\ldots,X_n \}$.
Now, let $\mathcal{S} \subseteq 2^{[n]}$ be a set of possible subsets of observed variables which are correctly specified by $\mu$, the remaining ones being misspecified -- formally, we introduce the following definition:
\begin{definition}
\label{def:intro-S-robust}
An algorithm $\texttt{ALG}$ solves the $\mathcal{S}$-misspecified inference problem for $\mu$ with error $\delta$ if for every $S \in \mathcal{S}$ and every $\mu'$ such that $\mu |_S = \mu'|_S$,
\[
\EE_{x = (x_S, x_{\overline S}) \sim \mu'} d_{TV}(\texttt{ALG}(x), \{X_0 \, | \, X_S = x_S \}) \leq \delta \, .
\]
\end{definition}
Importantly, the algorithm $\texttt{ALG}$ only knows $\mathcal{S}$, not the particular $S$ or $\mu'$.

In Definition~\ref{def:intro-S-robust}, $\mu'$ is our ``ground truth'' description of the world, and $\mu$ is our misspecified model.
Given a sample $x_1,\ldots,x_n$ from $\mu'|_{\{1,\ldots,n\}}$, if we knew $\mu'$ then the best inference we could make about $x_0$ would be the conditional distribution of $x_0$ given $x_1,\ldots,x_n$.
But there is a problem -- we do not know $\mu'$.
If we knew $S$, we could discard the observations whose distributions we know nothing about, and infer $\{X_0 \, | \, X_S\}$.
We do not know $S$, but Definition~\ref{def:intro-S-robust} requires \texttt{ALG} to compete with a hypothetical algorithm that does.

Of course, many other possible notions of model misspecification for the posterior inference problem are possible -- we make no claim that this definition is universally applicable.
Exploring alternative notions of graphical model misspecification and their consequences for posterior inference is a fascinating open direction.

\paragraph{Interpreting our results as robustness to model misspecification}
In the context of the broadcast process on trees, we take $X_0$ above to be the sign at the root vertex $\sigma_R$, and $X_1,\ldots,X_n$ to be the leaf vertex signs $\sigma_L$.
Theorem~\ref{thm:lower-bound-intro} shows if we take $\mathcal{S}$ to be all subsets of size $(1-\rho) d^t$ then the $\mathcal{S}$-misspecified problem is impossible.
Theorem~\ref{thm:const-corrupt} has the following corollary, showing that there is a large set $\mathcal{S}$ of possible misspecifications against which robust inference is possible (indeed, a random $\mathcal{S}$ suffices).
\begin{corollary}
\label{cor:intro-misspecification}
    Under the same hypotheses on $\delta,d,\eps,\rho$ as Theorem~\ref{thm:const-corrupt}, there is some $t_0(d, \delta)$ such that if $t \geq t_0$ then there exists $\mathcal{S} \subseteq 2^{[d^t]}$ with $|\mathcal S| \geq 2^{\Omega_{\eps,d}(d^t)}$ and an algorithm which solves the $\mathcal{S}$-misspecified inference problem with error $\delta$, where $\Omega_{\eps,d}(\cdot)$ indicates asymptotic behavior as $t \rightarrow \infty$.
\end{corollary}
The proof of this corollary constructs a $\rho$-semi-random adversary which can sample from a misspecification distribution $\mu'$.
We defer the proof to Section~\ref{subsec:defer}.

\subsection{Related Work}\label{sec:rel-work}

\paragraph{Algorithmic Robust Statistics}
Algorithmic robust statistics in high dimensions has developed rapidly in both statistics and computer science in the last decade.
This field has focused on parameter estimation, distribution learning, and prediction using (otherwise-)iid samples of which a small fraction have been maliciously corrupted.
The recent book \cite{diakonikolas2023algorithmic} provides a comprehensive overview.

We highlight one parallel between our results and robust supervised learning.
In PAC learning, the \emph{Massart noise} model \cite{massart2006risk} is a parallel to our semirandom adversary.
In the Massart noise model, for each labeled example $(x,f(x))$, the learner gets to see $(x,y)$, where $y$ is selected by flipping a coin which comes up heads with probability $\rho < 1/2$ and then allowing an adversary who sees $x$ the chance to make a (randomized) decision whether to let $y = f(x)$ or $y = 1-f(x)$.
The Massart noise model is an important middle ground between randomized noise models (e.g. \emph{random classification noise} \cite{kearns1990computational}) and nastier noise models (e.g. agnostic learning), and has recently led to several algorithmic advances \cite{diakonikolas2019distribution,diakonikolas2020learning}.

A couple of recent works in algorithmic robust statistics design ``robustified'' versions of belief propagation or its dense-graph analogue, approximate message passing \cite{ivkov2023semidefinite,liu2022minimax}, often in the context of adversarially-robust algorithms analogues of the ``linearized-BP'' algorithm for community detection (often called the ``nonbacktracking spectral algorithm'' \cite{abbe2018community}) in the stochastic blockmodel \cite{banks2021local,ding2022robust}.

Finally, our work has a parallel in \cite{chen2022kalman}, which studies Bayesian inference with corruption in a Gaussian time-series setting (Kalman filtering).
This work also finds a strong information-theoretic impossibility result for corruptions at adversarially-chosen times (corresponding for us to adversarially-chosen leaf locations), and an efficient inference algorithm when corruptions are made at random times but with knowledge of the rest of the time series.

\paragraph{Broadcasting on Trees}
Broadcasting on trees is an important special case of the ferromagnetic Ising model, and the problem of reconstructing the root vertex label from observations at the leaves has been extensively studied in probability, statistical physics, and computer science \cite{mossel2004survey}.
It has significant connections to Markov-chain mixing times \cite{berger2005glauber,martinelli2004glauber}, phylogenetic reconstruction \cite{daskalakis2006optimal}, and community detection \cite{mossel2015reconstruction}.
Of particular interest because of an important application to inference in the stochastic block model \cite{MNS16,YP22} is the ``robust reconstruction'' problem, originally introduced in \cite{janson2004robust} to study sharpness of the Kesten-Stigum phase transition.
From the perspective of our work, this is the root-inference problem in the presence of a random adversary.

\paragraph{Robust Bayesian Inference}
Ensuring that Bayesian inferences are robust to inaccuracies in the choices of priors and models is a longstanding concern dating at least to the 1950s \cite{good1950probability,berger1986robust}, with too vast a literature to survey here.
Well-studied approaches involve choosing flat-tailed or noninformative priors and placing a hyperprior on a family of models to capture uncertainty about the true model.

\subsection{Overview of Proofs}\label{sub-sec:overview}


\paragraph{Optimal recovery from $\rho$-semirandom adversary}
In Section~\ref{sec:small-eps}, we prove Theorem~\ref{thm:const-corrupt} in the regime of small $\eps$, that is $\eps \leq \eps^{*}$ for some small $\eps^{*}$. In Section~\ref{subsec:large-eps}, we finish the proof of Theorem~\ref{thm:const-corrupt} by addressing the case that $\eps >\eps^*$. Both proofs use the following template.
The belief propagation algorithm produces, for each vertex $u$, a ``belief'' $Z_u \in [-1,1]$, which implicitly specifies an inferred posterior distribution $\{ \sigma_u \, | \, \sigma_{\text{decendents}(u)} \}$ by specifying the bias of that distribution.
$Z_u$ is a function of $Z_{u1},\ldots,Z_{ud}$, where $u1,\ldots,ud$ are children of $u$ and $Z_u = \sigma_u$ for a leaf-vertex $\sigma$.

Let $X_u$ denote the belief which would have been produced by BP if it were run on a non-corrupted leaf spins $\sigma_L$, and $Z_u$ the belief produced when BP is run with corrupted leaf spins.
Our goal will be to show that $|X_u - Z_u| < \delta$ for $\delta$ as small as we like, so long as $u$ is at large-enough height in the tree.

We begin by showing that when $d \eps^2$ exceeds Kesten-Stigum threshold by a $O(\log d)$ multiplicative factor, the difference $|Z_u - X_u|$ in our estimated belief $Z_u$ for any vertex $u$ at height $1$ is moderately close to the actual ``ground truth'' belief $X_u$. That is, we show that \[\EE \max_{\sf adversary}|Z_u - X_u| = O(\eps) \, , \]
where the expectation is taken over the randomness in the broadcast process and also the allowed corruption locations $x_u$ in the $\rho$-semirandom adversary (Lemma~\ref{lem:contract-small}), and $\max_{\sf adversary}$ denotes maximizing over choices of leaf-spin flips made by the $\rho$-semirandom adversary at the allowed locations $x$.

Next, we employ a contraction argument to show that this ``worst case'' perturbation of beliefs contracts as we move up the broadcast tree. More precisely, in Lemma~\ref{lem:contract-induct} we basically show that if $u$ is the parent of $ui$ then 
\[ \EE \max_{\sf adversary}|Z_u - X_u| \leq \frac{1}{2} \EE \max_{\sf adversary}|Z_{ui} - X_{ui}|. \]
This would then show that by taking a sufficiently large tree, we are able to recover the belief at the root to arbitrarily high precision. 

We note that \cite{MNS16} employs a similar contraction argument. However, the random adversary of \cite{MNS16} flips each leaf spin independently. In comparison, our $\rho$-semirandom adversary can introduce new long-range correlations between the leaves in our broadcast tree.

Our contraction argument uses a first order Taylor expansion of the belief propagation function. Fix a vertex $u$ with children $u1, \ldots, ud$. In Lemma~\ref{lem:contract-induct}, we effectively show that there is some $f \colon [-1,1] \to \RR$ which captures the effect of each $Z_{ui}$ on $BP(Z_{u1},\ldots,Z_{ud})$ in the sense that
\begin{align*}
    \EE \max_{\sf adversary}|Z_u - X_u| &= \EE \max_{\sf adversary} |BP(Z_{u1}, \ldots, Z_{ud}) - BP(X_{u1}, \ldots, X_{ud})| \\
    &\leq \EE \max_{\sf adversary} \left| \prod_{i=1}^{d} f(Z_{ui}) - \prod_{i=1}^{d} f(X_{ui}) \right| \\
    &\leq \EE \left[ \max_{\sf adversary}\sum_{i=1}^{d} \max \{ |f'(Z_{ui})|,  |f'(X_{ui})|\}  \cdot |X_{ui} - Z_{ui}| \cdot \max \left\{ \prod_{j \neq i}f(Z_{uj}), \prod_{j \neq i}f(X_{uj}) \right\} \right]
\end{align*}
where the second inequality above follows from an application of the mean value theorem (plus some additional facts about monotonicity of $f$).

Now, if $|X_{ui} - Z_{ui}|$ were independent of $Z_{uj},X_{uj}$ (conditioned on $\sigma_u$), as it would be in the random adversary setting, we could modify the above chain of inequalities to condition each one on $\sigma_u$ and then use this independence to separate the $|X_{ui} - Z_{ui}|$ term from the $\max \{ \prod_{j \neq i} f(Z_{uj}),\prod_{j \neq i} f(X_{uj}) \}$ term. As long as we could show that the latter term was small, we would be able to obtain a contraction in this way.

But because of long range correlations introduced by the adversary, even though $|X_{ui} - Z_{ui}|$ is expected to be small by the induction hypothesis, it is conceivable that the adversary can coordinate its flips in a way that blows up $\max \{ |f'(Z_{ui})|,  |f'(X_{ui})|\}  \max \left\{ \prod_{j \neq i}f(Z_{uj}), \prod_{j \neq i}f(X_{uj}) \right\} $.
To circumvent this, we carefully reintroduce independence by splitting the $\rho$-semi-random adversary into several \emph{independent} ``local'' adversaries who can only see small subtrees. Formally, this corresponds to continuing the above with the bound
\begin{align*}
    &\leq \EE \left[ \sum_{i=1}^{d} \left ( \max_{\sf adversary}\max \{ |f'(Z_{ui})|,  |f'(X_{ui})|\}  \cdot |X_{ui} - Z_{ui}| \right ) \cdot \left ( \max_{\sf adversary} \max \left\{ \prod_{j \neq i}f(Z_{uj}), \prod_{j \neq i}f(X_{uj}) \right ) \right\} \right]
\end{align*}

Strengthening the induction hypothesis to include $\EE \max_{\sf adversary}|Z_{uj} - X_{uj}| = O(\eps)$ will allow us to bound remaining terms above involving $f$ and $f'$, leading to our contraction.

\paragraph{Information-theoretic lower bound: Proof of Theorem~\ref{thm:lower-bound-intro}}
To prove that a $\rho$-fraction adversary makes posterior inference impossible, we construct a coupling between the distributions $\{ \sigma_L \, | \, \sigma_R = 1 \}$ and $\{ \sigma_L' \, | \, \sigma_R' = -1\}$ such that with all but exponentially-small probability, $\sigma_L$ and $\sigma_L'$ differ on a $d^{-\Omega(t)}$-fraction of coordinates.
The key observation is that, for a spin $\sigma_u$ and children $\sigma_{u1},\ldots,\sigma_{ud}$, the distributions $\{ \sigma_{u1},\ldots,\sigma_{ud} \, | \, \sigma_u = 1 \}$ and $\{ \sigma_{u1},\ldots, \sigma_{ud} \, | \, \sigma_u = -1 \}$ have Wasserstein distance roughly $\eps d$.
If we are allowed to flip an $\eps$-fraction of coordinates, we can couple the distributions successfully with high probability.
This $\eps$-fraction propagates down a height-$t$ tree to become an $\eps^t$ fraction of necessary flips.

This argument can even be adapted to the $\rho$-semirandom adversary when $\rho$ is $\Omega(\eps)$ (see Theorem~\ref{thm:semi-random-info}), by showing that the coupling needs only to flip signs which the $\rho$-semirandom adversary is allowed allowed to flip.



\paragraph{Spread adversary}

In Theorem~\ref{thm:well-spread} we show that for every $c>0$ there exists $k$ such that if an adversary makes at most $c$ corruptions in every height-$k$ subtree of the broadcast tree, then we can optimally infer the root. 

The algorithm for Theorem~\ref{thm:well-spread} proceeds in two stages: first we run a ``noise injection'' phase -- flipping each leaf node independently with small probability -- and then we run belief propagation. We appeal to the bounded-sensitivity property of the posterior inference function in Theorem~\ref{thm:c-flip} (which is where the noise injection phase arises) in order to establish the base case of the contraction. The inductive step of the contraction is identical to that of Theorem~\ref{thm:const-corrupt}.

\section{Preliminaries}\label{sec:prelim}
For a set $\Omega$, we write $\Delta(\Omega)$ to denote the space of probability distributions over $\Omega$. If $\Omega = \{ \pm 1\}$, a distribution $\nu \in \Delta(\{ \pm 1\})$ is also associated to a \emph{belief} $B = \EE_{b \sim \nu} b$.
For $\nu, \nu' \in \Delta(\{\pm 1\})$, clearly $d_{TV}(\nu,\nu') \leq O(|B - B'|)$ where $B,B'$ are the respective beliefs.

For a random variable $X$ and event $E$, we write $\{ X \, | \, E \}$ for the distribution of $X$ conditioned on $E$. For distributions $\mu, \nu$, we write $d_{TV}(\mu, \nu)$ for the total variation distance between $\mu$ and $\nu$. We write $\cD_{d, \eps, t}$ to denote the random process of generating spins on depth $t$ $d$-regular trees according to the broadcast on tree process with parameters $d, \eps$. In an abuse of notation, we will often write $(\sigma_R, \sigma_L) \sim \cD_{d, \eps, t}$ to mean a sample from the marginal distribution of $\cD_{d, \eps, t}$ on the spins of the root $R$ and the $d^t$ spins of the leaves $L$.

\section{$\rho$-semi-random Adversary}\label{sec:semi-rand}
In this section we prove our main result, Theorem~\ref{thm:const-corrupt}. We make no attempt at optimizing the relevant constants and prove Theorem~\ref{thm:const-corrupt} with
\begin{equation}
    \label{eq:parameter-bounds}
     t_0 =  \log(d) + \log(\delta^{-1}) \,, \,  \text{and}\quad \rho = \frac{\eps}{4}.
\end{equation}

Given the output of a broadcast on tree process that was run up till depth $t$, we recursively apply belief propagation until we reach the root and output the belief thus obtained. Specifically, suppose the spins of the leaves are given by $Z_{i, 0}$ ($1 \leq i \leq d^{t}$), then recursively define 
\[ Z_{j,i} = BP(Z_{dj+1, i-1}, Z_{dj+2, i-1}, \ldots, Z_{d(j+1)-1, i-1}) \]
for $1 \leq i \leq t$ and $1 \leq j \leq d^{i}$, where $BP \colon \RR^d \to \RR$ is the belief propagation function given by 
\[ BP(X_1, \ldots, X_d) = \frac{\prod_{i=1}^{d}(1 + \varepsilon X_{i}) - \prod_{i=1}^{d}(1 - \varepsilon X_{i})}{\prod_{i=1}^{d}(1 + \varepsilon X_{i}) + \prod_{i=1}^{d}(1 - \varepsilon X_{i})}. \] 

We argue separately in the cases of small and large $\eps$.

\subsection{Small-$\eps$ case}
\label{sec:small-eps}
In this section we will prove Theorem~\ref{thm:const-corrupt} when $\eps < \eps^*$ for some sufficiently small constant $\eps^*$; we capture this in Lemma~\ref{lem:small-eps}.
This small-$\eps$ case already captures the main ideas in the proof of Theorem~\ref{thm:const-corrupt}; we defer the large-$\eps$ case to Section~\ref{subsec:large-eps}.
Before we state Lemma~\ref{lem:small-eps} we introduce some convenient notation.
\begin{definition}
    Given $d, \varepsilon, \rho, t$, for a function $f \, : \, \{ \pm 1 \}^{d^t} \rightarrow \RR$, let $\EE \max_{\sf{adversary}} f(x)$ be the expected value of:
    \begin{itemize}
        \item Sampling $\sigma_L \sim \mathcal{D}_{d,\eps,t}$ and a $0/1$-valued vector $x$ of length $d^t$ with entries independently equal to $1$ with probability $\rho$ and $0$ otherwise
        \item Given $\sigma_L$ and $x$, return $\max f(\sigma_L')$ where $\sigma_L$ and $\sigma_L'$ differ only on the coordinates indicated by $x$.
    \end{itemize}
\end{definition}

The main result in this section is: 
 \begin{lemma}\label{lem:small-eps}
     There exist absolute constants $C$ and $\varepsilon^{*} > 0$ such that if $d\varepsilon^2 \geq C \log(d)$ and $\varepsilon \leq \varepsilon^{*}$ then for any $\delta > 0$,
     with $\rho(\eps,d), t_0(\delta,d)$ as in \eqref{eq:parameter-bounds}, for any $t \geq t_0$,
    \[\EE \max_{\sf{adversary}} {|X_{\sf{root}, t} - Z_{\sf{root}, t}|}  \leq \delta.\]
 \end{lemma} 

Note that Lemma~\ref{lem:small-eps} establishes Theorem~\ref{thm:const-corrupt} in the case $\eps \leq \eps^{*}$ since $X_{\sf root, t}$ is the bias corresponding to the posterior distribution of the root.
Before we describe the main ideas behind the proof, we need one more piece of notation.
If $x \in \{0,1\}^{d^t}$ is the random bitstring produced by the $\rho$-semi-random adversary, we define:
\begin{definition}
    Given a vector $x \in \{0,1\}^{d^t}$ indexed by leaves of a $(d+1)$-ary tree of depth $t$ and an (internal) node $u$ in that tree, we write $x_{u}$ to denote the restriction of $x$ to the leaves of the subtree rooted at $u$.
\end{definition}

The key steps in our contraction argument are given by the following two lemmas. The first, Lemma~\ref{lem:contract-small} captures that even at height-$1$ internal nodes, the adversary cannot (in expectation) corrupt the beliefs by more than $O(1/d)$.
We will prove Lemma~\ref{lem:contract-small} later in this section.

\begin{lemma}\label{lem:contract-small}
    There exists an absolute constant $C$ such that for any $k \geq 1$, if $d\varepsilon^2 \geq C \log\left( \tfrac{d}{1-\eps} \right)$ then for any vertex $u$ at level $t-1$, we have
    \[ \EE \max_{\sf{adversary}} \left[\left|X_{u, 1} - Z_{u, 1}\right| \mathrel{\Big|}  \sigma_u = +1 \right] \leq \frac{(1-\eps)^{1/4}}{100 d}.\] 
\end{lemma}

We note that because we are above the Kesten-Stigum regime, it follows that $\frac{(1-\eps)^{1/4}}{100 d} \leq \frac{\eps}{100}$, and so we can recover the beliefs at level $t-1$ up to error $O(\eps)$. 
The $O(\eps)$ bound on the RHS of Lemma~\ref{lem:contract-small} is crucial as the base case for an induction up the remaining $t-1$ levels of the tree.
Our induction hypothesis, captured below in Lemma~\ref{lem:contract-induct}, relies on the computed beliefs $Z_{u,r}$ at a higher level $r \geq 1$ being distance at most $\lesssim \eps$ to $X_{u,r}$ (in expectation).


For the both the base case and inductive step of the contraction argument, we will start by rearranging the belief propagation function definition of $X_{u,r}$ to $X_{u,r} = 2 / (1 + \prod_{i \leq d} \tfrac{1-\eps X_{ui,r-1}}{1+ \eps X_{ui,r-1}}) - 1$.
We obtain a simpler bound on $X_{ur} - Z_{ur}$ by then applying an elementary-calculus bound on $\tfrac{1}{1+x} - \tfrac{1}{1+y}$ to $|X_{u,r} - Z_{u,r}|$:
\begin{claim}[Proof in Section~\ref{subsec:defer}]
\label{claim:small-1}
    For any $0 < p < 1$ and any $x,y \geq 0$, we have $\left| \frac{1}{1+x} - \frac{1}{1+y} \right| \leq \frac{1}{p} |x^p - y^p|$. 
\end{claim}

So, it will be enough to bound $\left | \left( \prod_{i=1}^{d} \frac{1 - \varepsilon X_{ui,r}}{1 + \varepsilon X_{ui,r}} \right)^{1/2} - \left( \prod_{i=1}^{d} \frac{1 - \varepsilon Z_{ui,r}}{1 + \varepsilon Z_{ui,r}} \right)^{1/2} \right|$.
Conditioned on the values of $|X_{ui,r} - Z_{ui,r}|$, the extremal values for $\left( \prod_{i=1}^{d} \frac{1 - \varepsilon Z_{ui,r}}{1 + \varepsilon Z_{ui,r}} \right)^{1/2}$ are achieved when $\mathrm{sign}(X_{ui,r} - Z_{ui,r})$ is the same for all $i$; we introduce random variables $Y_{v,r} = X_{v,r} - |X_{v,r} - Z_{v,r}|$ and $Y'_{v,r} = X_{v,r} + |X_{v,r} - Z_{v,r}|$ to capture the hypothetical situation that all these signs have lined up.

The following Lemma~\ref{lem:contract-induct} says that both the resulting extremal values are close to $\left ( \prod_{i =1}^d \tfrac{1-\eps X_{ui,r}}{1 + \eps X_{ui,r}} \right )^{1/2}$.
The technical key underlying our contraction argument is that we get a bound relying only on the assumption that 
\[
\max_{i \leq d} \EE \left[\max_{\sf{adversary}}| Y_{ui,r} - X_{ui,r}|\mathrel{\Big|}  \sigma_{ui} = +1 \right]
\]
is in turn bounded.
It would be much easier to argue under the stronger assumption of a bound on
\begin{equation}
\label{eq:bad-assumption}
\EE \left [ \max_{i \leq d} \max_{\sf{adversary}}| Y_{ui,r} - X_{ui,r}|\mathrel{\Big|}  \sigma_{u} = +1 \right] \, .
\end{equation}
The latter would amount to the assumption that the adversary has been unable to have much effect on the beliefs computed at \emph{any} of the children of $u$.
But this would be too strong an assumption to use inductively -- the risk is that at level $r+1$ of the induction we would need an assumption on
\begin{equation}
\EE \left [ \max_{i,j \leq d} \max_{\sf{adversary}}| Y_{uij,r} - X_{uij,r}|\mathrel{\Big|}  \sigma_{u} = +1 \right] \, ,
\end{equation}
and so on.
This eventually would amount to a union bound aiming to show that $|Y_{u,k} - X_{u,k}| \leq \eps$ simultaneously for every height-$k$ vertex $u$, but this simply isn't true -- since the adversary gets to corrupt a constant fraction of leaves, there are subtrees of the broadcast tree where they could achieve, say, $|X_{u,k} - Y_{u,k}| > 0.1$.

In the non-adversarial setting, \cite{mossel2015reconstruction} can make a similar argument, avoiding an assumption like \eqref{eq:bad-assumption} by leveraging the independence of the beliefs $Z_{ui,k}$ conditioned on the sign $\sigma_u$.
Crucially for them, the signs of $Z_{ui,k} - X_{ui,k}$ conditioned on $\sigma_u$ are independent across $i=1,\ldots,d$, which leads important cancellations.
This independence fails in the adversarial setting, because the adversary's choices whether to corrupt a leaf vertex may depend on signs of other far-away leaves!
We carefully re-introduce independence by replacing the adversary who sees the whole tree with ``local'' adversaries who see only subtrees -- we must argue that these local adversaries are not too much weaker than the original $\rho$-semirandom adversary.
We show this argument in Section~\ref{subsec:small-eps-contract}.

\begin{lemma}\label{lem:contract-induct}
    Define $ Y_{v,r} = X_{v,r} - |X_{v,r} - Z_{v,r}| \, ,$ and $ Y_{v,r}' = X_{v,r} + |X_{v,r} - Z_{v,r}|$. There exist constants $\varepsilon^{*}, C', R > 0$ such that if $\varepsilon \leq \varepsilon^{*}$, $\varepsilon^2 d > C'$, the following holds: let $r \geq 1$, for any vertex $u$ at level $t-r-1$ with children $u1, \ldots, ud$ and suppose $\xi := \EE  \left[\max_{\sf{adversary}}| Y_{ui,r} - X_{ui,r}|\mathrel{\Big|}  \sigma_{ui} = +1 \right] \leq \frac{\eps}{10}$, then 
    \[ \EE\left[ \max_{\sf{adversary}} \left| \left( \prod_{i=1}^{d} \frac{1 - \varepsilon X_{ui,r}}{1 + \varepsilon X_{ui,r}} \right)^{1/2} - \left( \prod_{i=1}^{d} \frac{1 - \varepsilon Y_{ui,r}}{1 + \varepsilon Y_{ui,r}} \right)^{1/2} \right| \mathrel{\Big|}  \sigma_u = +1 \right] \leq d e^{-R(d-1)\eps^2} \xi. \] 
    Furthermore, the above holds with $Y'_{ui,k}, Y'_{ui,r}$ replacing $Y_{ui,k}, Y_{ui,r}$ respectively.
\end{lemma}

Given Lemmas~\ref{lem:contract-small} and~\ref{lem:contract-induct}, we show how to deduce Lemma~\ref{lem:small-eps}.


\begin{proof}[Proof of Lemma~\ref{lem:small-eps}]
   Let $\alpha = \eps/10$. We prove by induction on $r \geq 1$ that for any vertex $u$ at level $r$, we have 
   \begin{equation}\label{eq:non-asymp-small}
       \EE \left[ \max_{\sf{adversary}} |X_{u,t-r-1} - Z_{u,t-r-1}| \mathrel{\Big|} \sigma_u = +1 \right] \leq 2^{-(r-1)} \alpha.
   \end{equation}
   Having proved this, it will suffice to take $t_0 \asymp \log(\delta^{-1})$ and iteratively apply (\ref{eq:non-asymp-small}).

    The base case when $r= 1$ follows from Lemma~\ref{lem:contract-small}. Assume the claim is true for $r$ and we aim to prove the claim for $r+1$. Fix a vertex $u$ at level $t-(r+1)$ and let its children be $u1, \ldots, ud$. 

    
    By taking $p = \frac{1}{2}$ in Claim~\ref{claim:small-1}, it follows for any $k$ that 
\begin{align}\label{eq:polyrize}
    |X_{u,k+1} - Z_{u,k+1}| &= |BP(X_{u1, k}, \ldots, X_{ud,k}) - BP(Z_{u1,k}, \ldots, Z_{ud,k})| \notag \\
    &\leq 4 \left| \left( \prod_i \frac{1 - \varepsilon X_{ui,k}}{1 + \varepsilon X_{ui,k}} \right)^{1/2} - \left( \prod_i \frac{1 - \varepsilon Z_{ui,k}}{1 + \varepsilon Z_{ui,k}} \right)^{1/2} \right|.
\end{align}


By the monotonicity of $x \mapsto \tfrac{1- \eps x}{1+\eps x}$, we have the pointwise inequalities 
\[ \frac{1 - \varepsilon Y_{ui,k}'}{1 + \varepsilon Y_{ui,k}'} \leq \frac{1 - \varepsilon Z_{ui,k}}{1 + \varepsilon Z_{ui,k}} \leq \frac{1 - \varepsilon Y_{ui,k}}{1 + \varepsilon Y_{ui,k}} \]
and therefore 
\begin{align}\label{eq:monotone}
   &\left| \left( \prod_i \frac{1 - \varepsilon X_{ui,k}}{1 + \varepsilon X_{ui,k}} \right)^{1/2} - \left( \prod_i \frac{1 - \varepsilon Z_{ui,k}}{1 + \varepsilon Z_{ui,k}} \right)^{1/2} \right| \notag \\
   &\leq \max \left\{\left| \left( \prod_i \frac{1 - \varepsilon X_{ui,k}}{1 + \varepsilon X_{ui,k}} \right)^{1/2} - \left( \prod_i \frac{1 - \varepsilon Y_{ui,k}}{1 + \varepsilon Y_{ui,k}} \right)^{1/2} \right|, \left| \left( \prod_i \frac{1 - \varepsilon X_{ui,k}}{1 + \varepsilon X_{ui,k}} \right)^{1/2} - \left( \prod_i \frac{1 - \varepsilon Y_{ui,k}'}{1 + \varepsilon Y_{ui,k}'} \right)^{1/2} \right|  \right\}. 
\end{align}

That is, the ``worst case adversary'' is the one that is able to perturb the spins at level $t$ such that running belief propagation on these spins produces beliefs $\{ B_{ui,k}\}$ satisfying $\sgn(B_{ui,k} - Z_{ui,k}) = \sgn(B_{uj,k} - Z_{uj,k})$ for all $1 \leq i,j \leq k$. It follows that 
    \begin{align*}
        &\EE\left[\max_{\sf{adversary}}|X_{u,t -r -1} - Z_{u,t-r-1}| \mathrel{\Big |} \sigma_u = +1 \right]\\
        &\leq 4 \EE \left[ \max_{\sf{adversary}}  \left| \left( \prod_i \frac{1 - \varepsilon X_{ui,t-r}}{1 + \varepsilon X_{ui,t-r}} \right)^{1/2} -  \prod_i \left(\frac{1 - \varepsilon Z_{ui,t-r}}{1 + \varepsilon Z_{ui,t-r}} \right)^{1/2} \right|  \mathrel{\Big |} \sigma_u = +1 \right]  \\
        &\leq 4 \EE \left[ \max_{\sf{adversary}}  \left| \left( \prod_i \frac{1 - \varepsilon X_{ui,t-r}}{1 + \varepsilon X_{ui,t-r}} \right)^{1/2} -  \prod_i \left(\frac{1 - \varepsilon Y_{ui,t-r}}{1 + \varepsilon Y_{ui,t-r}} \right)^{1/2} \right|  \mathrel{\Big |} \sigma_u = +1 \right] \\
        &\qquad + 4 \EE \left[ \max_{\sf{adversary}}  \left| \left( \prod_i \frac{1 - \varepsilon X_{ui,t-r}}{1 + \varepsilon X_{ui,t-r}} \right)^{1/2} -  \prod_i \left(\frac{1 - \varepsilon Y_{ui,t-r}'}{1 + \varepsilon Y_{ui,t-r}'} \right)^{1/2} \right|  \mathrel{\Big |} \sigma_u = +1 \right]
    \end{align*}
    Applying Lemma~\ref{lem:contract-induct} to the above gives 
    \begin{align*}
          \EE\left[\max_{\sf{adversary}}|X_{u,t-r-1} - Z_{u,t-r-1}| \mathrel{\Big |} \sigma_u = +1 \right] &\leq 16d (2^{-r-1} \alpha) e^{-R(d-1)\eps^2}
    \end{align*}
    and by choosing $C$ sufficiently large such that $\eps^2 d > C \log(d)$, we can make the above smaller than $2^{-(r+1)-1}\alpha$ as desired.
\end{proof}

\subsubsection{Base case: level $t-1$}

In this subsection, we prove Lemma~\ref{lem:contract-small}. Fix a vertex $u$ at level $t-1$, and suppose the sum of the uncorrupted spins of its $d$ children is $S$ while the sum of the corrupted spins of its $d$ children is $S'$. 

\subparagraph{Proof idea: } The main idea is that a Chernoff bound gives $S, S' = \Theta(\eps d)$. Vaguely speaking, if most of the input to the belief function $BP$ function is $+1$, then we would expect the output of $BP$ to also be quite close to 1. By writing $|Z_{u,1} - X_{u,1}| \leq \left | \left( \prod_{i=1}^{d} \frac{1 - \varepsilon X_{ui,r}}{1 + \varepsilon X_{ui,r}} \right)^{1/2} - \left( \prod_{i=1}^{d} \frac{1 - \varepsilon Z_{ui,r}}{1 + \varepsilon Z_{ui,r}} \right)^{1/2} \right|$ \`a la Claim~\ref{claim:small-1}, we can convert $S, S' = \Omega(\eps d)$ and $\eps^2 d \gtrsim \log(d)$ to a $O(1/d)$ bound on the magnitude $|Z_{u,1} - X_{u,1}|$.  

\begin{proof}[Proof of Lemma~\ref{lem:contract-small}]
    Let $S$ denote the random variable given by the sum of the $d$ uncorrupted spins of the children of $u$ and let $S'$ be the sum of the $d$ corrupted spins of the children of $u$. Then we have that $\EE[S \, | \, \sigma_u = +1] = \eps d$, and in particular the Chernoff bound implies that for sufficiently large $d$ (which we can arrange by taking a sufficiently large $C$), we have
    \[ \PP\left[S \leq \frac{\eps d}{2} \, \mathrel{\big|} \sigma_u = +1 \right] \leq \exp(-\eps d/8) \leq (1 - \eps/2)^{d} \leq \frac{(1-\eps)^{1/4}}{800d}.\]
    By a Chernoff bound, since $\rho = \eps/4$ in (\ref{eq:parameter-bounds}), we also have for sufficiently large $d$,
    \begin{equation}\label{eq:prob-uncorrupt-maj}
    \mathbb{P}\left[ |x_{u}| \geq \frac{\varepsilon d}{3} \right] \leq \exp\left(-\Omega\left(\varepsilon d \right)\right) \leq (1 - \eps/2)^{d} \leq \frac{(1-\eps)^{1/4}}{800d}.
\end{equation}
    In particular, it follows that 
    \[
    \PP \left ( S \geq \frac{\eps d}{2} \text{ and }  \min_{\sf adversary} S' \geq \frac{\eps d}{6} \right ) \geq \left(1 - \frac{(1-\eps)^{1/4}}{800d} \right)^2 \geq 1 -\frac{(1-\eps)^{1/4}}{400d}.
    \]

    Next, by Claim~\ref{claim:small-1}, note that we can bound 
    \begin{align*}
        \EE \max_{\sf adversary}|Z_{u,1} - X_{u,1}| &\leq  2\EE \left[ \max_{\sf adversary}\left | \left( \prod_{i=1}^{d} \frac{1 - \varepsilon X_{ui,r}}{1 + \varepsilon X_{ui,r}} \right)^{1/2} - \left( \prod_{i=1}^{d} \frac{1 - \varepsilon Z_{ui,r}}{1 + \varepsilon Z_{ui,r}} \right)^{1/2} \right| \mathrel{\big |} S \geq \frac{\eps d}{2} \, , \, S' \geq \frac{\eps d}{6} \right] \\ &\quad \quad \quad  + 2 \cdot \frac{(1-\eps)^{1/4}}{400d}  \\
        &= 2\EE \left[\max_{\sf adversary} \left|\left( \frac{1 - \eps}{1+\eps} \right)^{S'/2} - \left( \frac{1 - \eps}{1+\eps}\right)^{S/2} \right| \mathrel{\big |} S \geq \frac{\eps d}{2} \, , \,  S' \geq \frac{\eps d}{6} \right] + \frac{(1-\eps)^{1/4}}{200 d} \\
        &\leq 2\left( \frac{1-\eps}{1+\eps} \right)^{\frac{\eps d}{4}} + 2\left( \frac{1-\eps}{1+\eps} \right)^{\frac{\eps d}{12}} +  \frac{(1-\eps)^{1/4}}{200 d} \\
        &\leq 4e^{-\Omega_C(\eps^2 d)} + \frac{(1-\eps)^{1/4}}{200 d} \\
        &\leq \frac{(1-\eps)^{1/4}}{100 d}
    \end{align*}
    by taking a sufficiently large $C$. 
\end{proof}

\subsubsection{Contraction in the presence of an adversary}
\label{subsec:small-eps-contract}

We turn to the proof of Lemma~\ref{lem:contract-induct}. In the proof, we relate $\left| \left( \prod_{i=1}^{d} \tfrac{1 - \varepsilon X_{ui,r}}{1 + \varepsilon X_{ui,r}} \right)^{1/2} - \left( \prod_{i=1}^{d} \tfrac{1 - \varepsilon Y_{ui,r}}{1 + \varepsilon Y_{ui,r}} \right)^{1/2} \right|$ with $|X_{ui,r} - Y_{ui,r}|$ (whose expected value is bounded by assumption) by applying the mean value theorem. Specifically, we write 
\begin{equation}\label{eq:MVT}
    \left| \left( \prod_{i=1}^{d} \frac{1 - \varepsilon X_{ui,r}}{1 + \varepsilon X_{ui,r}} \right)^{1/2} - \left( \prod_{i=1}^{d} \frac{1 - \varepsilon Y_{ui,r}}{1 + \varepsilon Y_{ui,r}} \right)^{1/2} \right| \leq \sum_{i=1}^{d}|X_{ui,r} - Y_{ui,r}| \max_{\substack{y_1 \in[Y_{u1,r}, X_{u1,r}]}} \frac{\partial f}{\partial x}(y_i) \prod_{j \neq i} f(y_j)
\end{equation}
where $f(x) = \left( \frac{1- \eps x}{1+ \eps x}\right)^{1/2}$. The following two claims allow us to bound the size of each of the terms in the derivative. We comment that the proof of Lemma~\ref{lem:contract-induct} strongly utilizes the independence of the locations in $x$ where a flip is allowed. This is because a priori we are only able to work with the expected value version of \eqref{eq:MVT}, and it is this independence that allows us to split the products and bound them term by term.  


\begin{claim}\label{claim:small-3}
   Let $f(x) = \left(\frac{1-x}{1+x} \right)^{1/2}$. There exists $c > 0$ such that if $|x| \leq c$, then we have $|f(x) | \leq 1 -  x + \frac{3x^2}{5}$.
\end{claim}

\begin{claim}\label{claim:small-4}
    There are $\varepsilon^{*}, \kappa > 0$ such that for all $\varepsilon < \varepsilon^{*}$, all $x \in [-2, 2]$ and all $a \in [-2, 2]$, we have 
    \[ \left|\frac{\partial}{\partial x} \left(\frac{1 - \varepsilon (x+a)}{1 + \varepsilon (x +a) } \right)^{1/2} \right| \leq \kappa. \]
\end{claim}

We defer the proofs of Claims~\ref{claim:small-3} and \ref{claim:small-4} to the Appendix. We also quote without proof the following lemma from \cite{MNS16}. 

\begin{lemma}[{\cite[Lemma 3.5]{MNS16}}]\label{lem:var-majority}
    For a vertex $u$ at level $s$, let $u_1, \ldots, u_{d^r}$ denote its $d^r$ children at level $s+r$. Let $S_r$ denote the sum of the spins $\sum_{i=1}^{d^r} \sigma_{u_i}$. Then 
    \[ \mathrm{Var}[ S_r \, | \, \sigma_u = +1] = (1-\eps)^2 d^r \frac{(\eps^2 d)^r - 1}{\eps^2d - 1}.\]
\end{lemma}

Using this second moment bound on the majority vote, we can deduce that the ground truth beliefs $X_{ui,k}$ are all close to 1 in the regime where we exceed the Kesten-Stigum threshold by a constant. 

\begin{lemma}\label{lem:36}
    There exists $C > 0$ such that for all $d, \eps$ with $\eps^2 d > C$ , for all $s \geq 1$, we have \[\EE[X_{u,s} \, | \, \sigma_{u} = +1] \geq \frac{7}{8}.\]
\end{lemma}
\begin{proof}
    Let $S$ be the sum of the spins of the children $s$ levels down from $u$. By the optimality of $\sgn(X_{u,s})$ as an estimate for $\sigma_u$ given the spins of children $s$ levels down, we have 
    \begin{align*}
        \frac{\EE|X_{u,s}| + 1}{2} &\geq \PP(\sgn(S) > 0) \\
        &\geq 1 - \frac{\mathrm{Var}[s \, | \, \sigma_u = +1]}{(\EE[S \, | \, \sigma_u =+1])^2} \\
        &= 1 - \frac{\eps^{2s}d^{2s}}{(\eps^{s}d^{s})^2(\eps^2 d - 1)} \\
        &\geq 1 - \frac{1}{C - 1}
    \end{align*}
    where in the second inequality we applied the Chebyshev inequality. Observe that 
    \[ \PP[X_{u,s} <0 \, | \, \sigma_u = +1]  \geq \frac{1-\EE[|X_{u,s}| \, | \, \sigma_u = +1]}{2}. \]
    Consequently, 
    \begin{align*}
        \EE[X_{u,s} <0 \, | \, \sigma_u = +1]  &\geq \EE[|X_{u,s}| \, | \, \sigma_u = +1] - 2\PP[X_{u,s} < 0 \, | \, \sigma_u = +1] \\
        &= 2 \EE[|X_{u,s}| \, | \, \sigma_u = +1] - 1 \\
        &\geq 1 - \frac{4}{C-1} \\
        &\geq \frac{7}{8}
    \end{align*}
    for sufficiently large $C$. 
\end{proof}


\begin{proof}[Proof of Lemma~\ref{lem:contract-induct}]
    We will prove the case of $Y_{ui,r}$; the $Y_{ui,r}'$ case follows by nearly identical reasoning. 
    
    Now, let $f(x) = \left(\frac{1 - \eps x}{1+\eps x} \right)^{1/2} $, let $g(y_1,\ldots,y_d) = \prod_{j \leq d} f(y_j)$, and let $g_i(y_1,\ldots,y_d) = \tfrac{\partial g}{\partial y_i}$.
     By applying the mean value theorem, we can write 
    \begin{align}\label{eq:independence}
        &\EE \left[\max_{\sf{adversary}} \left| g(X_{u1,r}, \ldots, X_{ud,r}) - g(Y_{u1,r}, \ldots, Y_{ud,r}) \right| \mathrel{\Big|} \sigma_u = +1 \right] \notag  \\
        &\leq \EE \left[ \max_{\sf{adversary}} \sum_{i=1}^{d} |X_{ui,r} - Y_{ui,r}| \cdot \max_{\substack{y_1 \in[Y_{u1,r}, X_{u1,r}] \\ \vdots \\ y_d \in [Y_{ud,r}, X_{ud,r}]}} \left| g_i(y_1, \ldots, y_d) \right| \mathrel{\Big|}  \sigma_u = +1 \right] \\
        &\leq \kappa \EE \left[ \sum_{i=1}^{d} \max_{\sf{adversary}} |X_{ui,r} - Y_{ui,r}| \cdot \max_{\sf{adversary}} \max_{\substack{y_1 \in[Y_{u1,r}, X_{u1,r}] \\ \vdots \\ y_d \in [Y_{ud,r}, X_{ud,r}]}} \prod_{\substack{j=1 \\ j \neq i}}^{d} f(y_j)  \mathrel{\Big|}  \sigma_u = +1 \right] \notag.
    \end{align}
    
    In the above, $\kappa$ is the constant from Claim~\ref{claim:small-4}. In the second inequality, we introduced the \emph{``local'' adversaries} who only looks at the subtree beneath $ui$. At this point, we note that $\max_{\sf{adversary}} |X_{ui,r} - Y_{ui,r}|$ only depends on the randomness  which we revealed in the subtree below $ui$, while \\
    $\max_{\sf{adversary}}\max_{\substack{y_1 \in[Y_{u1,r}, X_{u1,r}] \\ \vdots \\ y_d \in [Y_{ud,r}, X_{ud,r}]}} \prod_{\substack{j=1 \\ j \neq i}}^{d} f(y_j)$ only depends on the randomness that we reveal in the subtrees below $uj$ for $j \neq i$. Because of the tree structure, it follows that if we condition on $\sigma_u = +1$, then $\max_{\sf{adversary}}|X_{ui,r} - Y_{ui,r}|$ is independent of $\max_{\sf{adversary}} \max_{\substack{y_1 \in[Y_{u1,r}, X_{u1,r}] \\ \vdots \\ y_d \in [Y_{ud,r}, X_{ud,r}]}} \prod_{\substack{j=1 \\ j \neq i}}^{d} f(y_j)$. It is crucial that the derivative of the belief propagation function has this separation property, allowing us to leverage the independence to write:
    
\begin{align*}
    &\EE \left[\max_{\sf{adversary}} \left| \prod_{i=1}^{d} f(X_{ui,r}) - \prod_{i=1}^{d} f(Y_{ui,r}) \right| \mathrel{\Big|}  \sigma_u = +1 \right]  \notag \\
        &\leq \kappa \sum_{i=1}^{d} \EE\left[\max_{\sf{adversary}} |X_{ui,r} - Y_{ui,r}| \mathrel{\Big|}  \sigma_{ui} = +1 \right] \cdot \EE \left[ \max_{\sf{adversary}} \max_{\substack{y_1 \in[Y_{u1,r}, X_{u1,r}] \\  \vdots \\  y_d \in [Y_{ud,r}, X_{ud,r}]}} \prod_{\substack{j=1 \\ j \neq i}}^{d} f(y_j)  \mathrel{\Big|}  \, \sigma_u = +1 \right] \notag \\
        &\leq \kappa d \xi \EE \left[ \max_{\sf{adversary}} \max_{\substack{y_1 \in[Y_{u1,r}, X_{u1,r}] \\ \vdots \\  y_d \in [Y_{ud,r}, X_{ud,r}]}} \prod_{j=1}^{d-1} f(y_j)  \mathrel{\Big|}  \sigma_u = +1 \right].
\end{align*}
    where in the last inequality we used symmetry of the subtrees and also the fact that 
    \begin{align*}
        \EE\left[  \max_{\sf adversary} |X_{ui,k} - Y_{ui,k}| \mathrel{\Big |} \sigma_{ui} = +1 \right] = \EE \left[ \max_{\sf adversary}|X_{ui,k} - Y_{ui,k}| \mathrel{\Big |} \sigma_{ui} = -1 \right] = \EE\left[ \max_{\sf adversary} |X_{ui,k} - Y_{ui,k}| \right].
    \end{align*}
    
    This equality holds because $\max_{\sf adversary}|X_{ui,k} - Y_{ui,k}|$ measures the magnitude of the change in beliefs of the worst adversary, and there by flipping $+1$ to $-1$ and vice versa there is a coupling between the broadcast processes when $ui = +1$ and $ui = - 1$ which leaves this magnitude invariant. In particular, this implies that 
    \begin{align*}
        \EE\left[\max_{\sf{adversary}} |X_{ui,r} - Y_{ui,r}| \mathrel{\Big|} \sigma_{ui} = +1 \right] \leq \xi
    \end{align*}
    by the hypothesis of the claim. Next, we write 
   \begin{align*}
       \EE \left[ \max_{\sf{adversary}} \max_{\substack{y_1 \in[Y_{u1,r}, X_{u1,r}] \\ \vdots \\ y_d \in [Y_{ud,r}, X_{ud,r}]}} \prod_{j=1}^{d-1} f(y_j)  \mathrel{\Big|} \sigma_u = +1 \right] &\leq \EE \left[ \max_{\substack{y_1 \in[Y_{u1,r}, X_{u1,r}] \\ \vdots \\ y_d \in [Y_{ud,r}, X_{ud,r}]}} \prod_{j=1}^{d-1} \max_{\sf{adversary}}f(y_j)  \mathrel{\Big|} \sigma_u = +1 \right] \notag \\
       &= \prod_{j=1}^{d-1}   \EE\left[ \max_{\sf{adversary}} \max_{\substack{y_1 \in[Y_{u1,r}, X_{u1,r}] \\ \vdots \\ y_d \in [Y_{ud,r}, X_{ud,r}]}} f(y_j)  \mathrel{\Big|}  \sigma_u = +1 \right] \notag \\
       &\leq  \prod_{j=1}^{d-1} \EE \left[ \max_{\sf{adversary}} f(X_{uj,r} - |Y_{uj,r} - X_{uj,r}|) \mathrel{\Big|} \sigma_u = +1 \right]
   \end{align*}
   where the third equality follows because of independence of $x_{ui}$ and $x_{uj}$ for $i \neq j$ and the final inequality follows from the monotonicity of $f$ so that $\max_{x \in [a,b]}f(x) = f(a)$. Now we apply Claim~\ref{claim:small-3} to the each term in the above inequality to obtain 
   \begin{align*}
       &\EE \left[\max_{\sf{adversary}} \left| \prod_{i=1}^{d} f(X_{ui,r}) - \prod_{i=1}^{d} f(Y_{ui,r}) \right| \mathrel{\Big|} \sigma_u = +1 \right] \\
       &\leq \kappa d \xi  \prod_{j=1}^{d-1} \EE \left[ \max_{\sf{adversary}} \left(1 - \eps(X_{uj,r} - |Y_{uj,r} - X_{uj,r}|) + \frac{3 \eps^2(X_{uj,r} - |Y_{uj,r} - X_{uj,r}|)^2}{5} \right) \mathrel{\Big|} \sigma_u = +1 \right] \\
       &\leq \kappa d \xi  \prod_{j=1}^{d-1} \EE \left[  1 -\eps(X_{uj,r} - \max_{\sf{adversary}}|Y_{uj,r} - X_{uj,r}|) + \frac{3\eps^2\max_{\sf{adversary}}(X_{uj,r} - |Y_{uj,r} - X_{uj,r}|)^2}{5} \mathrel{\Big|}  \sigma_u = +1 \right] \\
       &\leq \kappa d \xi  \prod_{j=1}^{d-1} \Bigg(1-\eps \EE[X_{uj,r} | \sigma_u = +1  ] + \eps \EE\max_{\sf{adversary}} [| Y_{uj,r} - X_{uj,r}| | \sigma_u = +1 ]  \\
       &\hskip 10cm + \frac{3\eps^2 \EE\underset{\sf adversary}{\max} [ (X_{uj,r} - |Y_{uj,r} - X_{uj,r}|)^2 |  \sigma_u = +1 ]}{5} \Bigg) \\
       &\leq \kappa d \xi  \prod_{j=1}^{d-1} e^{-\eps \EE[X_{uj,r} | \sigma_u = +1  ] + \eps \EE\max_{\sf{adversary}} [| Y_{uj,r} - X_{uj,r}| | \sigma_u = +1 ] + \frac{3\eps^2 \EE\max_{\sf{adversary}} [ (X_{uj,r} - |Y_{uj,r} - X_{uj,r}|)^2 |  \sigma_u = +1 ]}{5}} \\
        &\leq \kappa d \xi  \prod_{j=1}^{d-1}  e^{-\eps \EE[X_{uj,r} | \sigma_u = +1  ] + \eps \xi + \frac{3\eps^2 \EE\max_{\sf{adversary}} [ (X_{uj,r} - |Y_{uj,r} - X_{uj,r}|)^2 | \sigma_u = +1 ]}{5}},
   \end{align*}
   where for the penultimate inequality we used the fact that $1 - x \leq e^{-x}$ for all $x \in \mathbb{R}$.

    To simplify this further, note that by Lemma~\ref{lem:36} we have
 \begin{align}\label{eq:symmetrize}
    \EE[X_{ui,r} | \sigma_u = + 1] &= \frac{1 + \varepsilon}{2} \EE[X_{ui,r} | \sigma_{ui} = +1] + \frac{1 - \varepsilon}{2} \EE[X_{ui,r} | \sigma_{ui} = -1] \notag \\
    &= \EE\left[ \frac{1+ \varepsilon}{2} \cdot X_{ui,r} - \frac{1 - \varepsilon}{2} \cdot X_{ui,r} \ | \  \sigma_{ui} = +1 \right] \notag \\
    &= \varepsilon \EE[X_{ui,r} \ | \ \sigma_{ui} = +1] \geq \frac{7\varepsilon}{8}. 
\end{align}
And furthermore $ \EE\max_{\sf{adversary}} [ (X_{uj,r} - |Y_{uj,r} - X_{uj,r}|)^2 \mathrel{\Big|} \sigma_u = +1 ] \leq  1 + 4\xi$, since we have both $|X_{uj,r}|^2 \leq 1$, $\EE \max_{\sf{adversary}}[2 X_{uj,r} |Y_{uj,r} - X_{uj,k}|  \mathrel{\Big|} \sigma_u = +1] \leq 2\EE\max_{\sf{adversary}}[|Y_{uj,r} - X_{uj,r}|  \mathrel{\Big|}  \sigma_u = +1] \leq  2\xi$ and $\EE\max_{\sf{adversary}} [| Y_{uj,r} - X_{uj,r}|^2 \mathrel{\Big|} \sigma_u = +1 ] \leq 2\xi$. This allows us to write 

\begin{align}\label{eq:end-condition-base}
     \EE \left[\max_{\sf{adversary}} \left| \prod_{i=1}^{d} f(X_{ui,r}) - \prod_{i=1}^{d} f(Y_{ui,r}) \right| \mathrel{\Big|}  \sigma_u = +1 \right] \notag \leq \kappa d \xi e^{ (d-1)\eps^2 \left( -\frac{7}{8} + \xi \eps^{-1} + \frac{3 + 12\xi}{5} \right)} 
\end{align}
as desired.
\end{proof}

\subsection{Large-$\eps$ case}
\label{subsec:large-eps}
In this subsection, we aim to prove the following result. 

\begin{lemma}\label{lem:contract-large}
     For any $0 < \varepsilon^{*}<1$, there is some $C(\eps^*) > 0$ such that for all $\varepsilon > \varepsilon^{*}$ and $d$ satisfying $\eps^2 d > -C \log(1-\eps)$, then with $\rho(\eps), t_0(\delta,d)$ as in \eqref{eq:parameter-bounds}, for any $t \geq t_0$,
    \[\EE \max_{\sf{adversary}} {|X_{\sf{root}, t} - Z_{\sf{root}, t}|}  \leq \delta.\]
    Here, the expectation $\EE$ is taken with respect to the underlying broadcast process which produces leaf spins $\sigma_L$, and the $\text{Ber}(\rho)$ variables $x$ indicating which leaf spins can be corrupted by the adversary.
 \end{lemma} 

The gist of the proof of Lemma~\ref{lem:contract-large} is similar to that of Lemma~\ref{lem:contract-small}, with the main caveat being that approximations such as Claim~\ref{claim:small-3} no longer work, and so we need to use other estimates to bound the belief propagation function. To that end, we first state a technical lemma that we need to bound the ``worst case'' derivative of the belief propagation function. 


\begin{lemma}[Effectively {\cite[Lemma 3.16]{MNS16}}]\label{lem:term-der}
    For any $0 < \varepsilon^{*} < 1$, there is some $d^{*}(\varepsilon^{*})$, some $\lambda = \lambda(\varepsilon^{*}) < 1$ such that for all $1 > \varepsilon \geq \varepsilon^{*}$, $d \geq d^{*}$ there exists $\nu(\eps, \eps^{*})$ such that if $|\xi| \leq \nu $ and $Y_i = \min\{\max\{ X_{ui,k} - \xi, -1 \}, 1 \}$ for $k \geq 1$,
    \[ \EE \left[ \sqrt{\frac{1 - \varepsilon Y_i}{1 + \varepsilon Y_i}}  \mathrel{\Big|} \sigma_u = +1 \right] \leq \lambda. \] 
    In fact, we can take $\nu(\eps, \eps^{*}) = \frac{(\eps^{*})^2 (1-\eps)^{1/4}}{8\sqrt[4]{2}}$.
\end{lemma}

Because the proof of this lemma is near verbatim that of \cite[Lemma 3.16]{MNS16}, we defer its proof to Subsection~\ref{subsec:defer}.

First, we prove an analogue of Lemma~\ref{lem:contract-induct} in the setting of large $\eps$. For technical reasons, in the following we define the truncated random variables
\[Y_{ui,k} =  \max \{-1,  X_{ui,k} - |X_{ui,k} - Z_{ui,k}| \},\] 
and 
\[ Y_{ui,k}' = \min \{ X_{ui,k} + |X_{ui,k} - Z_{ui,k}|, 1 \}.\]
These random variables effectively behave like their counterparts from before, in terms of recording the ``worst case'' adversary. In fact, we observe that this definition of $Y_{ui,k}$ and $Y_{ui,k}'$ is quite natural; $-1$ and $1$ are the extreme points that the adversary can perturb the beliefs. The reason why such a truncation is necessary is basically because when we apply equation \cref{eq:polyrize}, we need to ensure that $\frac{1 - \eps Y_{ui,k}}{1 + \eps Y_{ui,k}}$ avoids the singularity $-\eps^{-1}$ of $\frac{1-\eps x}{1+\eps x}$. 

\begin{lemma}\label{lem:contract-unif-large}
For any $0 < \varepsilon^{*}<1$, there is some $C(\eps^*) > 0$ such that for every $1> \eps> \eps^{*}$ and $d$ such that $\eps^2d > C$, the following holds: let $r \geq 0$, for any vertex $u$ at level $t-r-1$ with children $u1, \ldots, ud$ and suppose $\xi := \EE  \left[\max_{\sf{adversary}}| Y_{ui,r} - X_{ui,r}|\mathrel{\Big|}  \sigma_{ui} = +1 \right] \leq \frac{(1-\eps) \nu(\eps, \eps^{*})}{d}$, where $\nu$ is as in Lemma~\ref{lem:term-der}, then 
    \[ \EE\left[ \max_{\sf{adversary}} \left| \left( \prod_{i=1}^{d} \frac{1 - \varepsilon X_{ui,r}}{1 + \varepsilon X_{ui,r}} \right)^{1/2} - \left( \prod_{i=1}^{d} \frac{1 - \varepsilon Y_{ui,r}}{1 + \varepsilon Y_{ui,r}} \right)^{1/2} \right| \mathrel{\Big|}  \sigma_u = +1 \right] \leq \frac{\xi}{100}. \]

Furthermore, the above holds with $Y'_{ui,r}$ replacing $Y_{ui,r}$.
\end{lemma}

\begin{proof}[Proof of Lemma~\ref{lem:contract-unif-large}]

    First, by applying Markov's inequality on the hypothesis of the claim, note that for any $1\leq i \leq d$ we have 
   \begin{equation}\label{eq:prob-large}
       \PP\left[ \max_{\sf adversary} |Y_{ui,r} - X_{ui,r}| \geq \nu \mathrel{\Big |} \sigma_u = +1 \right] \leq \frac{1-\eps}{d}.
   \end{equation}
    Because of the independence of $x_{ui}$ from $x_{uj}$ for $i \neq j$, we have for any $I\subset[d]$ such that $I = \{ i_1, \ldots, i_m\}$,
    \[ \PP\left[ \max_{\sf adversary} |Y_{ui,r} - X_{ui,r}| \geq \nu \,  \forall i \in I \mathrel{\Big |} \sigma_u = +1  \right] = \prod_{i=1}^{m} \PP\left[ \max_{\sf adversary} |Y_{ui,r} - X_{ui,r}| \geq \nu \mathrel{\Big |} \sigma_u = +1 \right] \leq \left( \frac{1-\eps}{d}\right)^{|I|}.\]

    Let $f(x) = \sqrt{\frac{1-\eps x}{1 + \eps x}}$. By the independence of $x_{ui}$ from $x_{uj}$ for $i \neq j$, we apply the mean value theorem to obtain
    \begin{align*}
            &\EE \left[\max_{\sf{adversary}} \left| \prod_{i=1}^{d} f(X_{ui,r}) - \prod_{i=1}^{d} f(Y_{ui,r}) \right| \mathrel{\Big|}  \sigma_u = +1 \right]  \notag \\
        &\leq  \sum_{i=1}^{d} \EE\left[\max_{\sf{adversary}} |X_{ui,r} - Y_{ui,r}| \mathrel{\Big|} \sigma_{ui} = +1 \right] \cdot \frac{1}{(1-\eps)^{3/2}} \cdot  \prod_{\substack{j=1 \\ j \neq i}}^{d}  \EE \left[ \max_{\sf{adversary}} \max_{\substack{y_1 \in[Y_{u1,r}, X_{u1,r}] \\  \vdots \\  y_d \in [Y_{ud,k}, X_{ud,k}]}}f(y_j)  \mathrel{\Big|}  \sigma_u = +1 \right] \notag \\ \\
        &\leq \frac{d \cdot \xi}{(1-\eps)^{3/2}} \EE \left[ \prod_{j=1}^{d-1} \max_{\sf adversary} f(Y_{uj,r}) \mathrel{\Big |}  \sigma_u = +1\right]
    \end{align*}
    where the second inequality holds by monotonocity of $f$ and the symmetry of subtrees rooted at $u1, \ldots, ud$. For the first inequality, note that since $\left|\frac{df}{dx} \right|$ is convex on $[-1,1]$ with a minimum at $x = \frac{1}{2}$ and we truncated the random variables $Y_{ui,k}$ so that we always have $X_{ui,k}, Y_{ui,k} \in [-1,1]$, we can bound $\left|\frac{df}{dx} \right|$ by its value at the endpoints \[\max \left\{ \left|\frac{df}{dx}(-1) \right|, \left|\frac{df}{dx}(1)\right| \right\} = \max \left\{ \frac{\eps}{\sqrt{1-\eps}(\eps +1)^{3/2}}, \frac{\eps}{\sqrt{1+\eps}(1-\eps)^{3/2}} \right\} \leq \frac{1}{(1-\eps)^{3/2}}.\]
    For the first inequality, we also used the fact that 
    \begin{align*}
        \EE\left[\max_{\sf{adversary}} |X_{ui,r} - Y_{ui,r}| \mathrel{\Big|}  \sigma_{ui} = +1 \right] &= \EE\left[\max_{\sf{adversary}} |X_{ui,r} - Y_{ui,r}| \mathrel{\Big|} \sigma_{u} = +1 \right] \leq \xi
    \end{align*}
    which is fundamentally because $|X_{ui,r} - Y_{ui,r}|$ measures the worst-case \emph{magnitude} by which an adversary can perturb the beliefs. 
    
    Let $\cB_I$ be the event that $\max_{\sf adversary} |Y_{ui,r} - X_{ui,r}| \geq \nu $ for all $i \in I$ and $\max_{\sf adversary} |Y_{ui,r} - X_{ui,r}| \leq \nu$ for all $i \not \in I$. 
    
    By the law of total probability, 
    \begin{align*}
    &\EE \left[ \prod_{j=1}^{d-1} \max_{\sf adversary} f(Y_{uj,r}) \mathrel{\Big |}  \sigma_u = +1\right] \\
    &\leq  \sum_{I \subset [d]} \PP[\cB_I \, | \,  \sigma_u = +1 ] \EE \left[ \prod_{j=1}^{d-1} \max_{\sf adversary} f(Y_{uj,r}) \mathrel{\Big |} \cB_I \, , \, \sigma_u = +1\right] \\
    &\leq \sum_{I \subset [d]}\PP[\cB_I \, | \, \sigma_u = +1 ] \left( \frac{1+\eps}{1-\eps} \right)^{\frac{|I|}{2}} \EE \left[ \prod_{i \not \in I}\sqrt{\frac{1 - \varepsilon (\max \{-1, X_{ui,r} - \nu \})}{1 + \varepsilon (\max \{-1, X_{ui,r} - \nu \})}}  \mathrel{\Big|} \cB_I \, , \,  \sigma_u = +1 \right] \\
    &\leq \sum_{I \subset [d]}\PP[\cB_I \, | \,  \sigma_u = +1 ] \left( \frac{1+\eps}{1-\eps} \right)^{\frac{|I|}{2}}\prod_{i \not \in I} \EE \left[ \sqrt{\frac{1 - \varepsilon (\max \{-1, X_{ui,r} - \nu \})}{1 + \varepsilon (\max \{-1, X_{ui,r} - \nu \})}}  \mathrel{\Big|} \cB_I \,  , \, \sigma_u = +1 \right].
\end{align*}

Next, we use Lemma~\ref{lem:term-der} to bound $\EE \left[ \sqrt{\frac{1 - \varepsilon (\max \{-1, X_{ui,r} - \nu \})}{1 + \varepsilon (\max \{-1, X_{ui,r} - \nu \})}}  \mathrel{\Big|} \cB_I \, , \, \sigma_u = +1 \right]$. Now, we note that 
\begin{align*}
    &\EE \left[ \sqrt{\frac{1 - \varepsilon (\max \{-1, X_{ui,r} - \nu \})}{1 + \varepsilon (\max \{-1, X_{ui,r} - \nu \})}}  \mathrel{\Big|} \cB_I \, ,  \, \sigma_u = +1 \right] \\
    &= \EE \left[ \sqrt{\frac{1 - \varepsilon (\max \{-1, X_{ui,r} - \nu \})}{1 + \varepsilon (\max \{-1, X_{ui,r} - \nu \})}}  \mathrel{\Big|} | \max_{\sf adversary} |Y_{ui,k} - X_{ui,k}| < \nu \, ,  \, \sigma_u = +1 \right]\\
    &\leq \frac{\EE \left[ \sqrt{\frac{1 - \varepsilon (\max \{-1, X_{ui,r} - \nu \})}{1 + \varepsilon (\max \{-1, X_{ui,r} - \nu \})}}  \mathrel{\Big|}  \sigma_u = +1 \right]}{\PP[\max\limits_{\sf adversary} |Y_{ui,r} - X_{ui,r}| < \nu  \, | \,  \sigma_{u} = +1 ]} \\
    &\leq \frac{\lambda}{1 - \frac{1}{d}} \\
    &\leq \widetilde{\lambda},
\end{align*}
where the second equality follows since $x_{ui}$ is independent from $x_{uj}$ for $i \neq j$, the penultimate inequality follows by applying Lemma~\ref{lem:term-der} to bound the numerator and applying \eqref{eq:prob-large} to bound the denominator and the last equality follows for some $\widetilde{\lambda}< 1$ by taking sufficiently large $C(\eps^{*})$ so that $d > C(\eps^{*})$ is sufficiently large. Putting all this together, it follows that 
\begin{align*}
    &\EE \left[ \prod_{j=1}^{d-1} \max_{\sf adversary} f(Y_{uj,r}) \mathrel{\Big |}  \sigma_u = +1\right]  \\
    &\leq \sum_{m=0}^{d-1} \binom{d-1}{m}\left( \frac{1-\eps}{d}\right)^{m} \left( \frac{2}{1-\eps} \right)^{\frac{m}{2}} \widetilde{\lambda}^{d- m-1} \\
    &\leq \left(\widetilde{\lambda} + O\left(\frac{1}{d} \right)\right)^{d-1}
\end{align*}
for sufficiently large $d > d^{*}$. Combining these bounds, we obtain 
\begin{align*}
     \EE \left[\max_{\sf{adversary}} \left| \prod_{i=1}^{d} f(X_{ui,r}) - \prod_{i=1}^{d} f(Y_{ui,r}) \right| \mathrel{\Big|}  \sigma_u = +1 \right] \leq \frac{d\left(\widetilde{\lambda} + O\left(\frac{1}{d} \right)\right)^{d-1}}{(1-\eps)^{3/2}} \cdot \xi \leq \frac{\xi}{100}
\end{align*}
by taking sufficiently large $C(\eps^{*})$ to ensure that \[\left( \frac{\lambda}{1 - \frac{1}{d}} + O\left( \frac{1}{d}\right)\right)^{d-1} \geq d^{-1}  \left( \frac{\lambda}{1 - \frac{1}{d}} + O\left( \frac{1}{d}\right)\right)^{\frac{d}{2}} \geq d^{-1}  \left( \frac{\lambda}{1 - \frac{1}{d}} + O\left( \frac{1}{d} \right)\right)^{-\frac{C}{4} \log(1- \eps)} \geq \frac{(1-\eps)^{3/2}}{100 d},\] as desired.
\end{proof}

With this inductive contraction lemma in place, the proof of Lemma~\ref{lem:contract-large} is immediate.

\begin{proof}[Proof of Lemma~\ref{lem:contract-large}]

    Let $\nu = \frac{(1-\eps)\nu(\eps, \eps^{*})}{d}$. We induct on $r \geq 1$ to prove that 
    \begin{equation}\label{eq:non-asymp-large}
       \EE \left[ \max_{\sf{adversary}} |X_{u,t-r} - Z_{u,t-r}| \mathrel{\Big|} \sigma_u = +1 \right] \leq 2^{-(r-1)} \nu.
   \end{equation}
   The base case of $r=1$ is Lemma~\ref{lem:contract-small} by taking $C$ sufficiently large, and the inductive step follows by applying Lemma~\ref{lem:contract-unif-large}. 
\end{proof}

\section{$(c,k)$-spread adversary}\label{sec:well-spread}

Finally, we introduce a deterministic adversary - the $(c,k)$-spread adversary - for which we can also accurately reconstruct the spin of the root.  

\begin{definition}
    Let $c >0$ and $k >0$ be given. The \emph{$(c,k)$-spread adversary} is an adversary that receives leaf signs $\sigma_L$ and is allowed to flip $c$ of the leaf spins of his choice among every height $k$ subtree. 
\end{definition}

\begin{theorem}\label{thm:well-spread}
    There exists a universal constant $C > 0$ with the following property. For any $c > 0$ and $d, \varepsilon > 0$ such that $\varepsilon^2 d > C\log(d)$, there exists some $k(c)$ and $t_0$ such that if $t>t_0$ then there exists an algorithm \texttt{ALG} which takes as input an element of $\{ \pm 1 \}^{d^t}$ and outputs an element of $\Delta(\pm1)$ with the property that for any $(c,k)$-spread adversary \texttt{A} acting on the broadcast, then 
    \[ \EE_{(\sigma_R,\sigma_L) \sim \cD_{d, \eps, t} \, , \, x} d_{TV}(\{ \sigma_R \, | \, \sigma_L \}, \texttt{ALG}(A_{\rho}(\sigma_L))) \leq \delta \, .
\]
\end{theorem}

As before, we implement a contraction argument. Fix $d, \eps$. Note, however, that for instance when $ c \geq d$, we cannot expect the base case of our contraction argument as in Lemma~\ref{lem:contract-small} to work because the adversary could choose to concentrate all of his flips in one subtree of height 1. For the base case of the induction, we instead look at a subtree of height $k$; this is captured by the following Theorem~\ref{thm:c-flip}.

\begin{definition}
    The \emph{$c$-flip adversary} is an adversary that is allowed to flip $c$ of the leaves of his choice after witnessing the entire broadcast tree process. 
\end{definition}

\begin{theorem}\label{thm:c-flip}
     There exists $C>0$ such that the following holds: for every $c,d,\varepsilon, \delta >0$ there exists an algorithm \texttt{A} and a constant $K(c) > 0$ such that if $d \varepsilon^2 > C$ and if $t \geq \log(K \delta^{-1})$, there exists an algorithm $\texttt{ALG}\colon \{\pm 1\}^{d^t} \to \Delta(\pm1)$ such that for any $c$-flip adversary \texttt{A}, we have 
     \[ \EE_{(\sigma_R, \sigma_L) \sim \cD_{d, \eps, t}} d_{TV}(\{ \sigma_R \, | \, \sigma_L \}, \texttt{ALG}(\texttt{A}(\sigma_L))) \leq \delta. \]
\end{theorem}

\begin{proof}[Proof of Theorem~\ref{thm:c-flip}]

Let $\psi = \min \{ \delta/(8c), \log(1+\delta/4)/(4c)\} $. To define the algorithm \texttt{A}, consider first the following random process $\mathcal{N}$. A draw $\sigma_L \sim \mathcal{N}$ corresponds to the output $\{ \pm 1 \}^{(d-1)^t}$ from running the broadcast on tree process with flip probability $\frac{1 - \varepsilon}{2}$ down a $d$-regular tree (i.e., each non-leaf node has $d-1$ children) and then passing each bit of the leaves through independent binary symmetric channels with flip probability $\frac{1 - \psi}{2}$.

We also define the ``clean'' process ${\cT}$: a draw $\sigma_L \sim {\cT}$ corresponds to the $\{\pm 1\}^{(d-1)^t}$ spins of the leaves at level $t$ of a broadcast on tree process with flip probability $\frac{1 - \varepsilon}{2}$. 

We let
\[
\texttt{A}_\cN(x) = \PP_{\mathcal{N}}(\sigma_R = 1 \mathrel{\Big|} \sigma_L = x) - \PP_{\mathcal{N}}(\sigma_R = -1 \mathrel{\Big|} \sigma_L = x) \, ,
\]
and
\[
\texttt{A}_{\cT}(x) = \PP_{{\cT}}(\sigma_R = 1 \mathrel{\Big|} \sigma_L = x) - \PP_{{\cT}}(\sigma_R = -1 \mathrel{\Big|} \sigma_L = x) \, ,
\]
both of which can be computed by belief propagation.
Ultimately, we will take $\texttt{A} = \texttt{A}_{\cN}$.

The goal is to find some $K$ such that for $t \geq \log(K\delta^{-1})$, we have the bound 
\begin{equation}
\label{eq:c-flip-posterior}
    \EE_{\sigma_L \sim {\cT}} \max_{v_1, \ldots, v_c \in L} \EE_{x \sim \Ber(\frac{1-\psi}{2})^{(d-1)^t}} \left|\texttt{A}_\cN \left(\sigma_L \oplus x \oplus \bigoplus_{i=1}^{c} e_{v_i} \right) - \texttt{A}_{\cT}(\sigma_L) \right| \leq \delta.
\end{equation}

We claim that \cite[Lemmas 3.4, 3.5]{MNS16} gives the quantitative bound that there is some constant $K(c) > 0$ such that if $t \geq \log(K\delta^{-1})$, then 
\begin{equation}\label{eq:abs-noisy}
    \EE_{\sigma_L \sim \mathcal{N}}|\texttt{A}_{\mathcal{N}}(\sigma_L)| \geq 1 - \frac{4}{\varepsilon^2 d}.
\end{equation}
Define $S$ to be the sum of all the spins drawn from the process $\mathcal{N}$. Then \cite[Lemmas 3.4, 3.5]{MNS16} gives 
\begin{align*}
    \begin{cases}
        \EE[S | \sigma_{\sf root} = +1] = \psi \varepsilon^t d^t, \\
        \mathrm{Var}[S | \sigma_{\sf root} = +1] = (1 - \psi^2) d^t + (1-\varepsilon^2) \psi^{2} \frac{((\eps^2d)^t - 1)d^t}{\eps^2d - 1}.
    \end{cases}
\end{align*}

Since $\psi \asymp_c \delta$, there exists $K(c)$ such that if $t \geq \log(K\delta^{-1})$, then $(\eps^2 d)^t > \frac{\psi^{-2}}{2}$, meaning that $(1 - \psi^2) d^t$ is dominated by $(1-\varepsilon^2) \psi^{2} \frac{((\eps^2d)^t - 1)d^t}{\eps^2d - 1}$.
In particular, in this case we have the (crude) bound of 
\[ \mathrm{Var}[S | \sigma_{\sf root} = +1] \leq \frac{3}{2}(1-\varepsilon^2) \psi^2 \frac{((\eps^2d)^t - 1)d^t}{\eps^2d - 1}.\]

By Chebyshev's inequality, it follows that if $t \geq \log(K\delta^{-1})$ then 
\[ \PP[S >0 | \sigma_{\sf root} = +1] \geq 1 - \frac{2}{\eps^2 d}. \]
Because $\sgn(\texttt{A}_{\cN}(x))$ is the optimal estimator of $\sigma_{\sf root}$ given $\sigma_L$, we have 
\begin{align*}
     \frac{1 + \EE|\texttt{A}_{\cN}(x)|}{2}&=\PP[\sgn(\texttt{A}_{\cN}(x)) = +1| \sigma_{\sf root} = +1]   \\
    &\geq \PP[\sgn(S) = +1| \sigma_{\sf root} = +1] \\
    &\geq 1 - \frac{2}{\eps^2 d}
\end{align*}
which upon rearranging gives \eqref{eq:abs-noisy}.

Combining \eqref{eq:abs-noisy} with the proof of \cite[Theorem 3.3]{MNS16}, we have that if $t \geq \log(K\delta^{-1})$, then 
\[ \EE_{\sigma_L \sim {\cT}} \EE_{x \sim \Ber(\frac{1 - \psi}{2})^{(d-1)^t}}|\texttt{A}_{\mathcal{N}}(x \oplus \sigma_L) - \texttt{A}_{{\cT}}(\sigma_L)| \leq \frac{1}{2} \EE_{\sigma_{L} \sim {\cT_{t-1}}} \EE_{x \sim \Ber(\frac{1 - \psi}{2})^{(d-1)^{t-1}}}|\texttt{A}_{\mathcal{N}}(x \oplus \sigma_{L}) - \texttt{A}_{{\cT_{t-1}}}(\sigma_L)|. \]
By iterating this contraction result, it follows that if $t \geq \log( \max\{2, K \} \delta^{-1})$, then we have 
\[ \EE_{\sigma_L \sim {\cT}} \EE_{x \sim \Ber(\frac{1-\psi}{2})^{(d-1)^t}}|\texttt{A}_\cN(x \oplus \sigma_L) - \texttt{A}_{\cT}(\sigma_L)| \leq \delta/2. \]

By the triangle inequality, in order to obtain the desired conclusion it suffices therefore to prove 
\[\EE_{\sigma_L \sim {\cT}} \max_{v_1, \ldots, v_c \in L} \EE_{x \sim \Ber(\frac{1-\psi}{2})^{(d-1)^t}} \left|\texttt{A}_\cN\left(\sigma_L \oplus x \oplus \bigoplus_{i=1}^{c} e_{v_i} \right) - \texttt{A}_\cN(x \oplus \sigma_L) \right| \leq \delta/2.\]

By symmetry and the triangle inequality, it suffices to prove that 
\begin{equation}\label{eq:lem1}
    \EE_{\sigma_L' \sim {\cT}}\max_{v_1, \ldots, v_c \in L} \EE_{x \sim \Ber(\frac{1-\psi}{2})^{(d-1)^t}} \left| \mathbb{P}_{\mathcal{N}}\left(\sigma_R = 1 | \sigma_L = \sigma_L' \oplus x \oplus \bigoplus_{i=1}^{c} e_{v_i} \right) - \mathbb{P}_{\mathcal{N}}(\sigma_R = 1 | \sigma_L = \sigma_L' \oplus x)  \right| \leq \delta/4.
\end{equation}

In fact, we will show that the difference above is at most $\delta/4$ for \emph{any} choice of $\sigma_L', v_1,\ldots,v_c, x$.
For any $x_L \in \{\pm 1\}^{(d-1)^t}$, using Bayes' rule, we can write 
\[ \PP_{(\sigma_R,\sigma_L) \sim \cN}[\sigma_R = 1\mathrel{\Big|}\sigma_L = x_L] = \frac{\PP_{(\sigma_R,\sigma_L) \sim \cN} [\sigma_L = x_L|\sigma_R = 1] \PP_{(\sigma_R,\sigma_L) \sim \cN}[\sigma_R = 1]}{\PP_{(\sigma_R,\sigma_L) \sim \cN}[\sigma_L = x_L]} \]
and for any choice of $v_1, \ldots, v_c$,
\[ \PP_{(\sigma_R,\sigma_L) \sim \cN}\left[\sigma_R = 1\mathrel{\Big|}\sigma_L = x_L \oplus \bigoplus_{i=1}^{c} e_{v_i} \right] = \frac{\PP_{(\sigma_R,\sigma_L) \sim \cN} \left[\sigma_L = x_L \oplus \bigoplus_{i=1}^{c} e_{v_i}|\sigma_R = 1 \right] \PP_{(\sigma_R,\sigma_L) \sim \cN}[\sigma_R = 1]}{\PP_{(\sigma_R,\sigma_L) \sim \cN}[\sigma_L = x_L \oplus e_v]}. \]
Note that $(\frac{1 - \psi}{2})^{c} \leq \frac{\PP_{(\sigma_R,\sigma_L) \sim \cN}\left[\sigma_L = x_L \oplus \bigoplus_{i=1}^{c} e_{v_i} \mathrel{\Big|}\sigma_R = 1 \right]}{\PP_{(\sigma_R,\sigma_L) \sim \cN}[\sigma_L = x_L|\sigma_R = 1]} \leq (\frac{1 + \psi}{2})^{c} $ and $(\frac{1 - \psi}{2})^{c} \leq \frac{\PP_{(\sigma_R,\sigma_L) \sim \cN}\left[\sigma_L = x_L \oplus \bigoplus_{i=1}^{c} e_{v_i} \right]}{\PP[\sigma_L = x_L]} \leq (\frac{1 + \psi}{2})^c$. That is, 
\begin{equation}\label{eq:psi-bound}
    1 - 2\psi c \leq \left( \frac{\frac{1 - \psi}{2}}{\frac{1 + \psi}{2}} \right)^c \leq \frac{\PP_{(\sigma_R,\sigma_L) \sim \cN}\left[\sigma_R = 1\mathrel{\Big|}\sigma_L = x_L \oplus \bigoplus_{i=1}^{c} e_{v_i} \right]}{\PP_{(\sigma_R,\sigma_L) \sim \cN}[\sigma_R = 1\mathrel{\Big|}\sigma_L = x_L]} \leq \left( \frac{\frac{1+\psi}{2}}{\frac{1-\psi}{2}} \right)^c \leq e^{4\psi c}.
\end{equation}

Our choice of $\psi$ gives the desired inequality \cref{eq:lem1}.
\end{proof}

As mentioned earlier, we use Theorem~\ref{thm:c-flip} as a primitive to kickstart our contraction argument. To be more precise, our algorithm and the relevant parameters in our proof of Theorem~\ref{thm:well-spread} are as follows. Fix $d,\eps, \delta$ and let $K(c)$ be the constant from Theorem~\ref{thm:c-flip} and set 
\begin{equation}
    \label{eq:parameter-bounds-1}
    \quad k = K(c) \, , \, t_0 = 10^5(\log(d) - \log(1-\eps) + \log(\delta^{-1})).
\end{equation}

Given the output of a broadcast on tree process that was run up till depth $t$, we partition the bottom $k$ levels of the tree into subtrees of size $d^{k}$ (for an illustration, refer to Figure~\ref{fig:tree}).
\begin{figure}[!htp]
    \centering
    \includegraphics[width = 8.5cm]{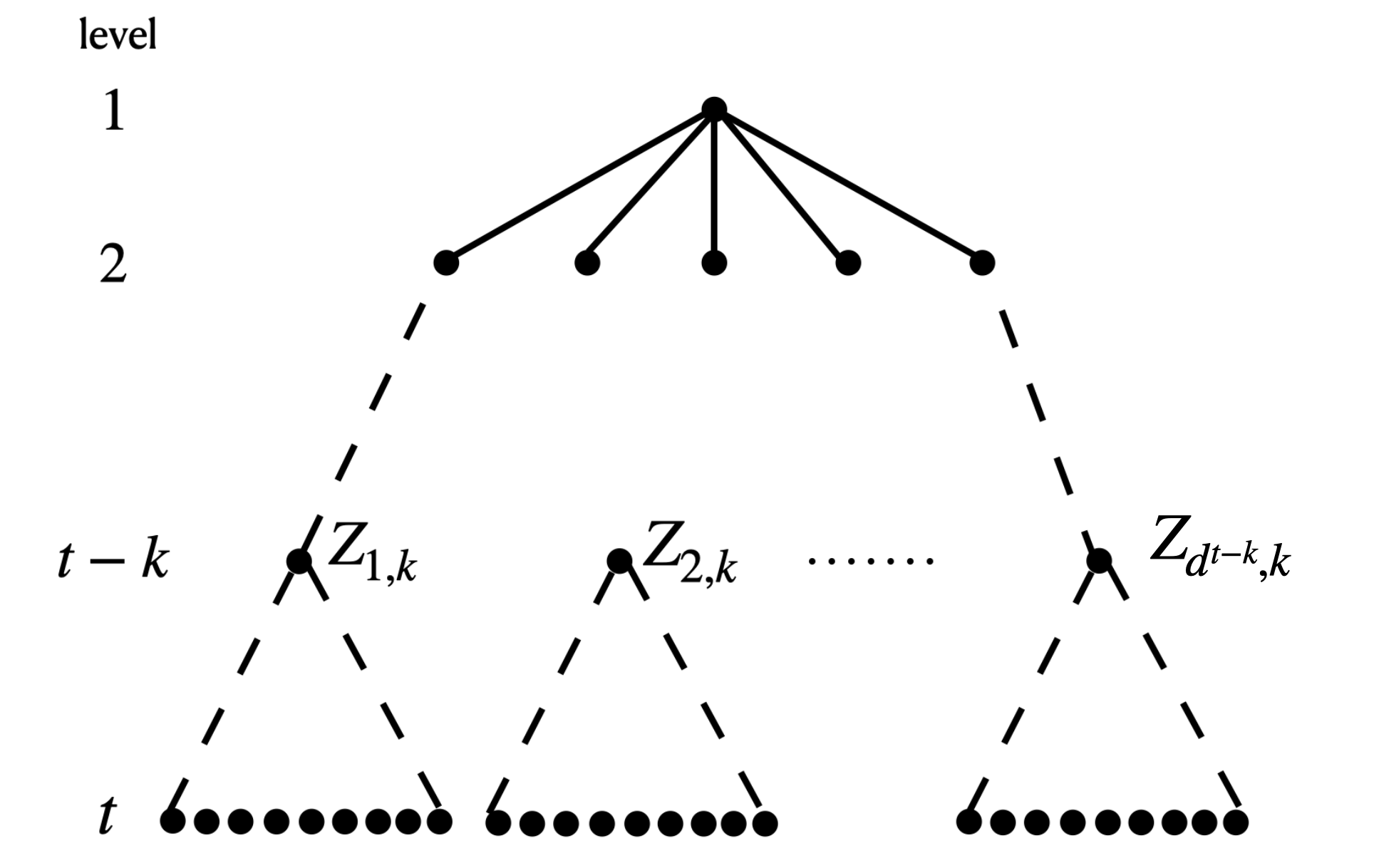}
    \caption{An illustration of the partition of the levels of the tree in the algorithm.}
    \label{fig:tree}
\end{figure}
On the leaves of each of these subtrees, run the algorithm from Theorem~\ref{thm:c-flip} and suppose the outputs are beliefs $Z_{i, k}$ ($1 \leq i \leq d^{t - k}$) for the $d^{t - k}$ nodes at level $t - k$ of the broadcast tree. Now, recursively apply belief propagation until we reach the root and output the belief thus obtained. Specifically, recursively define 
\[ Z_{j,i} = BP(Z_{dj+1, i-1}, Z_{dj+2, i-1}, \ldots, Z_{d(j+1)-1, i-1}) \]
for $k \leq i \leq t$ and $1 \leq j \leq d^{i}$, where $BP \colon \RR^d \to \RR$ is the belief propagation function given by 
\[ BP(X_1, \ldots, X_d) = \frac{\prod_{i=1}^{d}(1 + \varepsilon X_{i}) - \prod_{i=1}^{d}(1 - \varepsilon X_{i})}{\prod_{i=1}^{d}(1 + \varepsilon X_{i}) + \prod_{i=1}^{d}(1 - \varepsilon X_{i})}. \] 

We can restate Theorem~\ref{thm:c-flip} in terms of an upper bound on the maximum perturbation from the ground truth $X_{u,k}$. 

\begin{lemma}[Consequence of Theorem~\ref{thm:c-flip} for the parameters in (\ref{eq:parameter-bounds-1})]\label{lem:contract-small-spread}
      There exists an absolute constant $C$ such that if $d\varepsilon^2 \geq C \log(d)$ then for any vertex $u$ at level $t-k$ where $t \geq t_0$ where $k$ and $t_0$ are as in \eqref{eq:parameter-bounds-1}, we have
    \[ \EE \max_{\sf{adversary}} \left[\left|X_{u, k} - Z_{u, k}\right| \mathrel{\Big|}  \sigma_u = +1 \right] \leq \frac{\nu(\eps, \eps^{*})}{10 d}.\] 
\end{lemma}

\begin{proof}[Proof of Theorem~\ref{thm:well-spread}]
 We induct on $r \geq 1$ to prove that 
    \begin{equation}\label{eq:non-asymp-spread}
       \EE \left[ \max_{\sf{adversary}} |X_{u,t-r} - Z_{u,t-r}| \mathrel{\Big|} \sigma_u = +1 \right] \leq 2^{-(r-k)} \nu.
   \end{equation}
   The base case of $r=k$ is Lemma~\ref{lem:contract-small-spread}, and the inductive step follows by applying Lemma~\ref{lem:contract-induct} when $\eps \leq \eps^{*}$ and Lemma~\ref{lem:contract-unif-large} when $\eps > \eps^{*}$. 
   To finish up, it suffices to take $t_0 \asymp \log(\delta^{-1})$ and iteratively apply (\ref{eq:non-asymp-spread}).
\end{proof}

\section{Information-theoretic lower bounds}\label{sec:info-theory}

First, we establish the folklore result that the $\rho$-fraction adversary is very powerful and can in fact completely erode any information in the leaves about the posterior distribution of the root. 

\begin{definition}
    Fix $\rho >0$. The \emph{$\rho$-fraction adversary} is allowed to look at all the $d^t$ spins of the leaves and then flip the signs of $\rho d^t$ leaves of his choosing. 
\end{definition}

\begin{theorem}\label{thm:non-spread}
    For every $\rho$, $d$ and $\varepsilon$, there exists some $t_0$ such that if the $\rho$-fraction adversary is allowed to corrupt the leaves of the broadcast on tree process with parameters $(d, \varepsilon)$ that was run up to level $t \geq t_0$, it is information theoretically impossible to recover the root vertex. In fact, we prove that for $(\sigma_R, \sigma_L) \sim \cD_{d, \eps t}$ there exists a $\rho$-fraction adversary \texttt{A} for which we have
    \[
    d_{TV}( \{ \sigma_L \, | \, \sigma_R = 1 \}, \{ A(\sigma_L) \, | \, \sigma_R = -1 \}) \leq e^{-\Omega(t)} \,.
    \]
\end{theorem}

The quantitative bound in the theorem implies that it is impossible to have an algorithm \texttt{A} that recovers the root vertex with an advantage larger than $e^{-\Omega(t)}$.

\begin{proof}
    Set $t_0 = \log \rho^{-1}/(\log \frac{1+ 2 \varepsilon}{4 \varepsilon}) + 1$ and throughout this proof we consider $t \geq t_0$.
    Let $\mathcal{D}_t^{+}$ be the distribution of the leaves at level $t$ of a broadcast process with parameters $(\varepsilon, d)$ with root being $+1$ and $\mathcal{D}_t^{-}$ be the distribution with root being $-1$.
    And, let $D_{\leq t}^+, D_{\leq t}^-$ be the analogous joint distributions on the first $t$ levels of the broadcast tree.
    
    We claim that it suffices to exhibit a coupling $(\mathbf{x}, \pi_t(\mathbf{x})) \in \{ \pm 1 \}^{(d-1)^t} \times \{ \pm 1 \}^{(d-1)^t} $ between $\mathbf{x} \sim \mathcal{D}_t^{+}$ and $\pi_t(\mathbf{x})\sim \mathcal{D}_t^{-}$ such that with probability $1 - e^{-\Omega(t)}$, we have $\mathrm{dist}(\mathbf{x}, \pi_t(\mathbf{x})) \leq (4e \varepsilon)^{t} (d-1)^t$ where $\mathrm{dist}$ refers to Hamming distance between the two strings. 
     
    Indeed, consider the following $\rho$-fraction adversary $\mathbf{A}$: for $\mathbf{x} \sim \mathcal{D}_t^{+}$, if $\mathrm{dist}(\mathbf{x}, \pi_t(\mathbf{x})) \leq (4e \varepsilon)^{t} (d-1)^t$ then let $\mathbf{A}(\mathbf{x}) = \pi_t(\mathbf{x})$, where by our choice of parameters we have that $(4e \varepsilon)^t \leq \rho$. 
    Otherwise, let $\mathbf{A}(\mathbf{x}) = \mathbf{x}$. Let $\mathcal{A}_t^{+}$ denote the distribution corresponding to $\mathbf{A}(\mathbf{x})$ where $\mathbf{x} \sim \mathcal{D}_t^{+}$. In particular, this means that 
     \[ d_{TV}(\mathcal{D}_t^{-}, \mathcal{A}_t^{+}) \leq e^{-\Omega(t)}. \]
     
    It suffices therefore to exhibit the coupling $\pi_t$ as claimed.
    Define $\pi_t(\mathbf{x})$ for $\mathbf{x} \sim \mathcal{D}_{t}^{+}$ to act as follows: (in the following, let $\xi = \frac{4\varepsilon}{1+2 \varepsilon}$) here one should think of marked vertices as the vertices where the adversary did not make any changes, and the goal of the adversary is to flip unmarked vertices that are labelled `+' to `-'

    \begin{itemize}
        \item Sample $\mathbf{y} \sim \mathcal{D}_{\leq t}^{+} \mathrel{\Big|} \mathbf{x}$ (that is, $\mathcal{D}_{\leq t}^+$ conditioned on the leaves of the tree taking labels $\mathbf{x}$).
        \item Let $\mathbf{y}'(\sf{root}) = -1$.
        \item Traverse down the tree starting from the the nodes at level 1, that is, from the children of the root. If the current node $v$ is marked, let $\mathbf{y}'(v) = \mathbf{y}(v)$ and mark all its children. Otherwise, if $v$ is unmarked and $\mathbf{y}(v) = -1$, then mark $v$ and all its children and set $\mathbf{y}'(v) = \mathbf{y}(v)$. Else, if $v$ is unmarked and $\mathbf{y}(v) = +1$, with probability $\xi$ leave it unmarked and set $\mathbf{y}'(v) = -1$. Otherwise with probability $1 - \xi$, mark $v$ and its children and set $\mathbf{y}'(v) = +1$.
        \item Let $\pi_t(\mathbf{x})$ be the restriction of $\mathbf{y}'$ to the vertices at level $t$. 
    \end{itemize}
In order to prove that this coupling has the desired properties, we first show that with high probability, the Hamming distance between $\mathbf{x}$ and $\pi_t(\mathbf{x})$ is small.

Note that if for a vertex $v$ at the $t$-th level, we have $\pi_t(\mathbf{x}(v)) \neq \mathbf{x}(v)$ then $v$ has to be unmarked, and furthermore there is a path of $t$ unmarked vertices connecting $v$ to the root.
The latter occurs with probability at most $\xi^t$. Applying Markov's inequality immediately gives  

\[ \mathbb{P}[\mathrm{dist}(\mathbf{x}, \pi_t(\mathbf{x})) \geq \rho^t(d-1)^t] \leq e^{-\Omega(t)}.\]

Next, we need to show that $\pi_t(\mathbf{x}) \sim \mathcal{D}_t^{-}$. 
We will actually show that $\mathbf{y'} \sim \mathcal{D}^-_{\leq t}$ by inducting on $t$. For the base case, note that for any vertex $w$ on level 1, we have that 
\[ \mathbb{P}[\mathbf{y}'(w) = -] = \left( \frac{1}{2} - \varepsilon \right) + \left( \frac{1}{2} + \varepsilon \right) \xi = \frac{1}{2} + \varepsilon\]
as desired. 

For $s \leq t$, we write $\mathbf{y'}_{\leq s}$ to denote the restriction of $\mathbf{y'}$ to labels of vertices in levels $r \leq s$.
Note that it suffices to prove that for any vertex $v$ at level $s+1$ with parent $w$, 
\[ \mathbb{P}[\mathbf{y}'(v) = + | \mathbf{y}'_{\leq s}] = \frac{1}{2} + \mathbf{y}'(w)\varepsilon.\]
Indeed, note that 
\begin{align*}
    & \mathbb{P}[\mathbf{y}'(v) = + \ | \ \mathbf{y}'_{\leq s}]\\
    & \quad = \mathbb{P}[\mathbf{y}'(v) = + \ | \ \mathbf{y}'_{\leq s}, w \text{ marked}] \mathbb{P}[w \text{ marked} \ | \ \mathbf{y}'_{\leq s}]  +\mathbb{P}[\mathbf{y}'(v) = + \ | \ \mathbf{y}'_{\leq s}, w \text{ unmarked}] \mathbb{P}[w \text{ unmarked} \ | \ \mathbf{y}'_{\leq s}]\\
    & \quad = \left( \frac{1}{2} +  \mathbf{y}'(w)\varepsilon \right) \mathbb{P}[w \text{ marked} \ | \ \mathbf{y}'_{\leq s}] + \mathbb{P}[\mathbf{y}'(v) = + \ | \ \mathbf{y}'(w) = -, w \text{ unmarked}] \mathbb{P}[w \text{ unmarked} \ | \ \mathbf{y}'_{\leq s}].
\end{align*}
Since 
\begin{align*}
    \mathbb{P}[\mathbf{y}'(v) = + \ | \ \mathbf{y}'(w) = -, w \text{ unmarked}] &= 1 - \mathbb{P}[\mathbf{y}'(v) = - \ | \ \mathbf{y}'(w) = -, \mathbf{y}(w) = +, w \text{ unmarked}]  \\
    &= 1 - \left( \frac{1}{2} - \varepsilon + \left( \frac{1}{2} + \varepsilon \right) \xi \right) = \frac{1}{2} - \varepsilon \\
    &= \frac 12 + \mathbf{y}'(w) \varepsilon
\end{align*}

and $ \mathbb{P}[w \text{ marked} \ | \ \mathbf{y}'_{\leq s}] + \mathbb{P}[w \text{ unmarked} \ | \ \mathbf{y}'_{\leq s}] = 1$, the above indeed implies that 
\[ \mathbb{P}[\mathbf{y}'(v) = + | \mathbf{y}'_{\leq s}] = \frac{1}{2} + \mathbf{y}'(w)\varepsilon \]
as desired. 
\end{proof}

 Another natural question is whether the bound on $\rho$ in Theorem~\ref{thm:const-corrupt} is tight; that is, can we prove any information theoretic lower bounds on $\rho$-semi-random adversaries? To that end, we have the following upper bound of $\rho \lesssim \varepsilon$. This means that at least in the range of $d\eps^2 \gtrsim \log \tfrac{d}{1-\eps} $ being a log-factor above the Kesten-Stigum threshold, belief propagation achieves the optimal $\rho$ (up to constant factors) in terms of being robust against $\rho$-semi-random adversaries.  
\begin{theorem}\label{thm:semi-random-info}
    For every $d$ and $\varepsilon$, there exists $t_0$ such that if a $\left( \frac{4\varepsilon}{1+ 2 \varepsilon} \right)$-semi-random adversary is allowed to corrupt the leaves of the broadcast on tree process with parameters $(d,\varepsilon)$ that was run up to level $t \geq t_0$, then it is information theoretically impossible to recover the root vertex. That is, for $(\sigma_R, \sigma_L) \sim \cD_{d, \eps, t}$ there exists a $\left( \frac{4\varepsilon}{1+ 2 \varepsilon} \right)$-semi-random adversary \texttt{A} such that
    \[d_{TV}(\{\sigma_L \, | \, \sigma_R = 1 \}, \{\texttt{A}(\sigma_L) \, | \, \sigma_R = -1 \}) \leq e^{-\Omega(t)} \,.\]
\end{theorem}
\begin{proof}
    Set $t_0 = 1/(\log(2e\varepsilon)^{-1})$. We adapt the adversary that we used to prove the information theoretic bound in Theorem~\ref{thm:non-spread}. Retain the notations as in the proof of Theorem~\ref{thm:non-spread}, where $\pi_t$ defines a coupling on $(\mathcal{D}_t^{+}, \mathcal{D}_t^{-})$. Consider the following adversary \textbf{B}: for $\mathbf{x} \sim \mathcal{D}_t^{+}$, if all the sites of $\pi_t(\mathbf{x})\triangle \mathbf{x}$ (i.e. the sites where the two vectors differ) are selected as possible sites for corruption, and we call this event \emph{good}, then output $\pi_t(\mathbf{x})$. Otherwise, return $\mathbf{x}$. Let $\mathcal{B}_t^{+}$ denote the distribution corresponding to $\mathbf{B}(\mathbf{x})$ where $\mathbf{x}\sim \mathcal{D}_t^{+}$.
    
    We claim that the good event happens with probability $1$. To that end, we traverse down the tree and mark vertices as in the process described in Theorem~\ref{thm:non-spread} until we reach level $t-1$. For each of these marked vertices on level $t-1$, we simulate the process of choosing whether to mark its children with the coin flips in the $\left( \frac{4\varepsilon}{1+ 2 \varepsilon} \right)$-semi-random adversary. In other words we couple random coin flips to mark the children and change their spins with the coin tosses in the adversary. The bound on the semi-random robust gain follows from Theorem~\ref{thm:non-spread} as well. 
\end{proof}

\section*{Acknowledgements}
We thank Guy Bresler and Po-Ling Loh for helpful conversations as this manuscript was being prepared. 

\appendix 

\section{Deferred proofs}\label{subsec:defer}

\begin{proof}[Proof of Corollary~\ref{cor:intro-misspecification}]
    We claim that an equivalent formulation of Theorem~\ref{thm:const-corrupt} is that for any $\delta',d,\eps, \rho$ satisfying suitable hypothesis, there exists $t_0'(d, \delta')$ such that for any adversary \texttt{B} that takes as input $\sigma_L$ and $S \subset [n]$ and is only allowed to flip the signs of the restriction $\sigma_L|_{S}$, there exists an algorithm \texttt{ALG} with the following guarantee:
    \[ \EE_{(\sigma_R, \sigma_L) \sim \cD_{d, \eps, t}, S \sim \Delta_{\rho,t}} d_{TV}\left( \{ \sigma_R \, | \, \sigma_L \}, \texttt{ALG}(\texttt{B}(\sigma_L, S))\right) \leq \delta' \]
    for all $t \geq t_0'$ where $\Delta_{\rho,t}$ is the uniform distribution on subsets $S \subset [d^t]$ of cardinality $\frac{\rho d^t}{2}  \leq |S| \leq 2 \rho d^t$. 

    To this end, note that if $k \leq m$ then 
    \begin{multline*}
        \EE_{(\sigma_R, \sigma_L) \sim \cD_{d, \eps, t}, S \sim \Delta_{\rho,t}} \max_{|S| = k} d_{TV}\left( \{ \sigma_R \, | \, \sigma_L \}, \texttt{ALG}(\texttt{B}(\sigma_L, S))\right) \\ \leq \EE_{(\sigma_R, \sigma_L) \sim \cD_{d, \eps, t}, S \sim \Delta_{\rho,t}} \max_{|S| = m} d_{TV}\left( \{ \sigma_R \, | \, \sigma_L \}, \texttt{ALG}(\texttt{B}(\sigma_L, S))\right).
    \end{multline*}
    Let $\Delta'_{\rho,t}$ be the distribution on subsets $S \subset [d^t]$ of cardinality $\rho d^t \leq |S| \leq 2 \rho d^t$ where $\mathbb{P}[S] = \rho^{|S|}(1-\rho)^{n-|S|}$. Note that $\rho^{m}(1-\rho)^{n-m}$ is decreasing on the range $ \rho d^t \leq m \leq 2 \rho d^t$ since $\rho < \frac{1}{2}$. Consequently,
    \begin{align*}
    \EE_{(\sigma_R, \sigma_L) \, , \, S \sim \Delta_{\rho,n}} d_{TV}\left( \{ \sigma_R \, | \, \sigma_L \}, \texttt{ALG}(\texttt{B}(\sigma_L, S))\right) &\leq \EE_{(\sigma_R, \sigma_L) \, , \, S \sim \Delta'_{\rho,n}} d_{TV}\left( \{ \sigma_R \, | \, \sigma_L \}, \texttt{ALG}(\texttt{B}(\sigma_L, S))\right) \\ 
    &\leq \EE_{(\sigma_R,\sigma_L)} d_{TV}(\{ \sigma_R \, | \, \sigma_L \}, \texttt{ALG}(A(\sigma_L))) + \mathbb{P}[|\mathbf{x}| \geq 2\rho n] \\
    &\leq \EE_{(\sigma_R,\sigma_L)} d_{TV}(\{ \sigma_R \, | \, \sigma_L \}, \texttt{ALG}(A(\sigma_L))) + e^{-\rho d^t /6},
    \end{align*}
    where the final inequality follows from the Chernoff bound. This implies that we can take $t_0'(d, \delta') = t(d, \delta + e^{-\rho d^t/6})$. 
    
    By an application of Markov's inequality, it follows that $0.99$  of subset $S \subset [n]$ have the property that 
    \begin{equation}\label{eq:cor-intro}
        \EE_{(\sigma_R, \sigma_L) \sim \cD_{d, \eps, t}} d_{TV}\left( \{ \sigma_R \, | \, \sigma_L \}, \texttt{ALG}(\texttt{B}(\sigma_L, S))\right) \leq \delta.
    \end{equation}
    We claim that we can take $\cS$ to consist of the subsets $S$ which satisfy \eqref{eq:cor-intro}. First we check that these subsets have the desired property. Indeed, it suffices to note that for any $\mu$ with $\mu|_S = \cD_{d, \eps, t} |_S$, we have 
    \begin{align*}
        \EE_{ \sigma_L \sim \mu|_L} d_{TV}(\texttt{ALG}(\sigma_L, S), \{ \sigma_R |  \sigma_L \}) &= \EE_{\sigma_L|_{S}} \EE_{\sigma'_L|_{\overline{S}}} d_{TV}\left(\texttt{ALG}(\sigma_L|_S, \sigma'_L|_{\overline{S}},S), \EE_{\sigma_L|_{\overline{S}}}\{ \sigma_R |  \sigma_L\} \right) \\
        &\leq \EE_{\sigma_L, \sigma'_L|_{\overline{S}}} d_{TV}\left(\texttt{ALG}(\sigma_L|_S, \sigma'_L|_{\overline{S}},S), \{\sigma_R \, | \, \sigma_L \}  \right)
    \end{align*}
    where the first inequality follows because of the convexity of the total variation distance and the second inequality follows by \eqref{eq:cor-intro}. 

    Finally, we show that $\cS$ is large in size. This follows because 
    \[ |\mathcal{S}| \geq 0.99 \binom{d^t}{\rho d^t} \geq \exp(\rho d^t) = 2^{\Omega_{\eps, d}(d^t)}\]
    where the second inequality is true for sufficiently large $t$. 
\end{proof}


\begin{proof}[Proof of Claim~\ref{claim:small-1}]
    Let $f(x) = \frac{1}{1+x}$ and $g(x) = x^p$. We claim that 
    \[ |f'(x)| \leq g'(x)/p = x^{p-1}.\]
    The desired inequality then follows from the fundamental theorem of calculus. Now, $|f'(x)| = (1+x)^{-2}$ and so when $x \geq 1$, we have $|f'(x)| \leq x^{-2} \leq x^{p-1}$ and when $x \leq 1$ we have $|f'(x)| \leq 1 \leq x^{p-1}$. 
\end{proof}


\begin{proof}[Proof of Claim~\ref{claim:small-3}]
    It suffices to note that for sufficiently small $x$, we have 
    \[ \left( \frac{1-x}{1+x} \right)^{1/2} \leq 1 - x + \frac{3x^2}{5}. \]
    Indeed, the above holds as for sufficiently small $x$, we have
\begin{align*}
    (1+x)\left( 1 - x + \frac{3x^2}{5} \right)^2 &=  1 - x + \frac{x^2}{5} + O(x^3) \geq 1-x,
\end{align*}
as desired.
\end{proof}

\begin{proof}[Proof of Claim~\ref{claim:small-4}]
    It suffices to note that for
    \[  \frac{d}{dx} \left( \frac{\alpha -x}{\beta +x} \right)^{1/2}  = \frac{-(\alpha + \beta)}{4(\alpha - x)^{1/2} (\beta + x)^{3/2}}. \]
    By setting $\alpha = 1 - \varepsilon a$ and $\beta = 1 + \varepsilon a$, it follows that for sufficiently small $\varepsilon$, it follows that the desired quantity is bounded above by a constant.
\end{proof}

\begin{proof}[Proof of Lemma~\ref{lem:term-der}]
    In what follows, denote $\eta := \frac{1-\varepsilon}{2}$. We adapt the proof of \cite[Lemma 3.16]{MNS16}. For simplicity of notation denote $f(x) = \sqrt{\frac{1- \varepsilon x}{1 + \varepsilon x}}$. By applying Markov's inequality to Lemma~\ref{lem:36}, it follows that for some $\tau  = \tau(\varepsilon^{*})$ to be chosen and sufficiently large $\varepsilon^2 d$ relative to $\tau$, we have 
\begin{equation*}
    \PP[X_{ui,k} \geq 1 - \tau\eta^{1/4} \, | \, \sigma_u = +1   ] \geq 1 - \tau \eta^{3/4} - \eta.
\end{equation*}
Since $Y_i \geq X_{ui,k} - |\xi|$, it follows that
\begin{align*}
     \PP\left[Y_i \geq 1- \tau \eta^{1/4} - |\xi| \, | \, \, \sigma_u = +1  \right] \geq \PP\left[X_{ui,k} - |\xi| \geq 1- \tau \eta^{1/4} - |\xi| \, | \, \, \sigma_u = +1  \right] \geq 1 - \tau \eta^{3/4} - \eta  =: \alpha.
\end{align*}

The remainder of the proof is near verbatim that of \cite[Lemma 3.16]{MNS16} and we repeat it here for completeness. Since $f(x)$ is decreasing in $x$, it follows that 
\[ \EE \left[ f(X) \, | \, \sigma_u = +1   \right] \leq f(s) \PP[X\geq s \, | \, \sigma_u = +1   ] + f(-1) \PP[X \leq s \, | \, \sigma_u = +1  ]\]
for any random variable $X$ supported on $[-1,1]$ and $s \in [-1, 1]$. Applying this for $s = 1 - \tau \eta^{1/4} - \xi$ and $X = Y_i$, the above probability estimate implies that 
\[ \EE [f(Y_i) \, | \, \sigma_u = +1  ] \leq f(1-\tau \eta^{1/4} - |\xi|) \alpha + f(-1)(1 - \alpha). \]

We claim that we can make each of the terms above bounded away from $\frac{1}{2}$ by taking $\tau(\eps^{*}) = \frac{(\eps^{*})^2}{8}$ and $\nu(\eps, \eps^{*}) = \frac{(\eps^{*})^2 \eta^{1/4}}{8}$. Note that this choice of $\nu(\eps, \eps^{*})$ satisfies $\nu = \tau \eta^{1/4}$.  

We can compute
\begin{align*}
    f(1 - \tau \eta^{1/4} - |\xi|)(1 - \tau \eta^{3/4} - \eta) &\leq \sqrt{\frac{2 \eta + \eps \tau \eta^{1/4} + \eps \xi}{2 - (2\eta + \eps \tau \eta^{1/4} + \eps |\xi|)}}\cdot (1 - \tau\eta^{3/4} - \eta) \\
    &\leq \left( \sqrt{\frac{\eta}{1 - \eta}} + \frac{1}{\sqrt{2 \eta}} (\eps\tau  \eta^{1/4} + \eps |\xi|) \right)\cdot (1 - \tau \eta^{3/4} - \eta) \\
    &\leq \sqrt{\eta(1-\eta)} + \frac{1}{\sqrt{2\eta}} \eps (\tau \eta^{1/4} + |\xi|)(1-\eta) \\
    &\leq \sqrt{\eta(1-\eta)} + \sqrt{2}\tau \eta^{1/4}\sqrt{1 - \eta}
\end{align*}
where the second inequality follows by Taylor expanding $\sqrt{x}{1-x}$ around $x = \eta$, and our choice of parameters ensure that $2(\eta + \eps \nu) \leq 2(\eta + \nu) < 1$ so that the derivative of $\sqrt{x}{1-x}$  is bounded. Finally, it can be checked that the above is bounded away from $\frac{1}{2}$ for our choice of parameters for our choice of $\tau$ and $\nu$ since $\eta \in \left(0, \frac{1 - \eps^{*}}{2} \right]$; in fact our choice of parameters ensures that $\sqrt{\eta(1-\eta)} + 2\tau \eta^{1/4}\sqrt{1 - \eta} \leq \frac{1}{2} - c(\eps^{*})$ for some $c(\eps^{*}) >0$. This follows because for $ 0 < \eps < 1$, we have that $\eta^{1/4}\sqrt{1 - \eta} \leq \frac{1}{2^{3/4}} + \frac{\eps}{2^{3/4}}$ and we also have $\sqrt{\eta(1-\eta)} \leq \frac{1}{2} - \frac{\eps^2}{4}$ and so we can write 
\[ \sqrt{\eta(1-\eta)} + 2 \tau \eta^{1/4} \sqrt{1 - \eta} \leq \frac{1}{2} - (\eps^{*})^2 \left( \frac{1}{4} - \frac{1}{4 \cdot 2^{3/4}} \right). \]

We can also compute 
\[ f(-1)(2\tau \eta^{3/4} + \eta) = 2 \tau \eta^{1/4}\sqrt{1 - \eta} + \sqrt{\eta(1 - \eta)}. \]
Note that $\eta(1-\eta)$ is bounded away from $\frac{1}{2}$ on the interval $\eta \in \left(0, \frac{1 - \eps^{*}}{2} \right]$ by our earlier calculation. 

It follows that there is some $d^{*}(\varepsilon)$ such that for all $d \geq d^{*}(\varepsilon^{*})$, there is some $\lambda(\varepsilon^{*}) < 1$ such that 
\[ \EE\left[\sqrt{f(Y_i)} \mathrel{\Big |} \sigma_u = +1 \right] \leq \lambda \]
as desired.
\end{proof}

\bibliographystyle{alpha}
\bibliography{ref.bib}

\end{document}